\documentclass[aps,prx,longbibliography,twocolumn]{revtex4-1}
\usepackage{times}
\usepackage{amsmath}
\usepackage{amsfonts}
\usepackage{amssymb}
\usepackage{graphicx}
\usepackage{physics}
\usepackage[colorlinks=true]{hyperref}
\usepackage{mathrsfs}
\usepackage{euscript}
\usepackage[colorlinks=true]{hyperref}
\usepackage{amsthm}
\usepackage{stmaryrd}
\usepackage{bbm}
\usepackage{bm}

\theoremstyle{definition}
\newtheorem{lemma}{Lemma}
\newtheorem{defini}{Definition}
\newtheorem{prop}{Proposition}

\hypersetup{
     colorlinks   = true,
     linkcolor    = blue,
     citecolor    = blue
}

\begin{document}

\title{Measurement-induced topological entanglement transitions in symmetric random quantum circuits}
\author{Ali Lavasani}
\author{Yahya Alavirad}
\author{Maissam Barkeshli}
\affiliation{Department of Physics, Condensed Matter Theory Center, University of Maryland, College Park, Maryland 20742, USA}
\affiliation{Joint Quantum Institute, University of Maryland, College Park, Maryland 20742, USA}

\begin{abstract}
  Random quantum circuits, in which an array of qubits is subjected to a series of randomly-chosen unitary operations, have provided key insights into the dynamics of many-body quantum entanglement. Recent work showed that interleaving the unitary operations with single-qubit measurements can drive a transition between high- and low-entanglement phases. We study a class of symmetric random quantum circuits with two competing types of measurements in addition to unitary dynamics. We find a rich phase diagram involving robust symmetry-protected topological, trivial, and volume law entangled phases, where the transitions are hidden to expectation values of any operator and are only apparent by averaging the entanglement entropy over quantum trajectories. In the absence of unitary dynamics, we find a purely measurement-induced critical point, which maps exactly to two copies of a classical 2D percolation problem. Numerical simulations indicate this transition is a tricritical point that splits into two critical lines in the presence of arbitrarily sparse unitary dynamics with an intervening volume law entangled phase. Our results show that measurements alone are sufficient to induce criticality and logarithmic entanglement scaling, and arbitrarily sparse unitary dynamics can be sufficient to stabilize volume law entangled phases in the presence of rapid yet competing measurements.
\end{abstract}

\maketitle

\section{Introduction}

Generic  unitary dynamics drive quantum many-body systems into highly entangled states characterized by volume-law scaling of subsystem entanglement entropies. When this dynamics is intercepted by rapid local measurements, individual quantum trajectories are expected to collapse into low entanglement states characterized by area-law scaling of subsystem entanglement entropies. Recently, it was discovered that, at least in a class of models, these two phases are separated by a scale-invariant ``critical point" at a finite measurement rate~\cite{skinner2019measurement,li2018quantum,chan19}. Several aspects of this transition and its generalizations have been studied recently~\cite{li2019measurement,gullan19purify,gullans2019scalable,ludwing19tensor,altman19error,Szyniszewski19weak,tang2020dmrg,ludwig20replica,deluca19fermion,Vasseur20tensor,altman19phasetheory,qi20blackhole,jed20numeric,vicari20ising,ashvin20meanfield,nahum2020entanglement}.

In the limit of infinitely rapid local measurements, the state of the system crucially depends on the choice of measurement basis. Assuming one measures only commuting single-qubit operators, the wave-function collapses into an unentangled trivial product-state. However, if one chooses to measure a set of \textit{stabilizer} operators that stabilize a topological or a symmetry protected topological (SPT) wave-function, the resulting state, despite having area-law scaling of entanglement as well, would be \textit{topologically distinct} from the product state~\cite{hastings11topologicalorder,chen2013symmetry}.

\begin{figure}
  \includegraphics[width=\columnwidth]{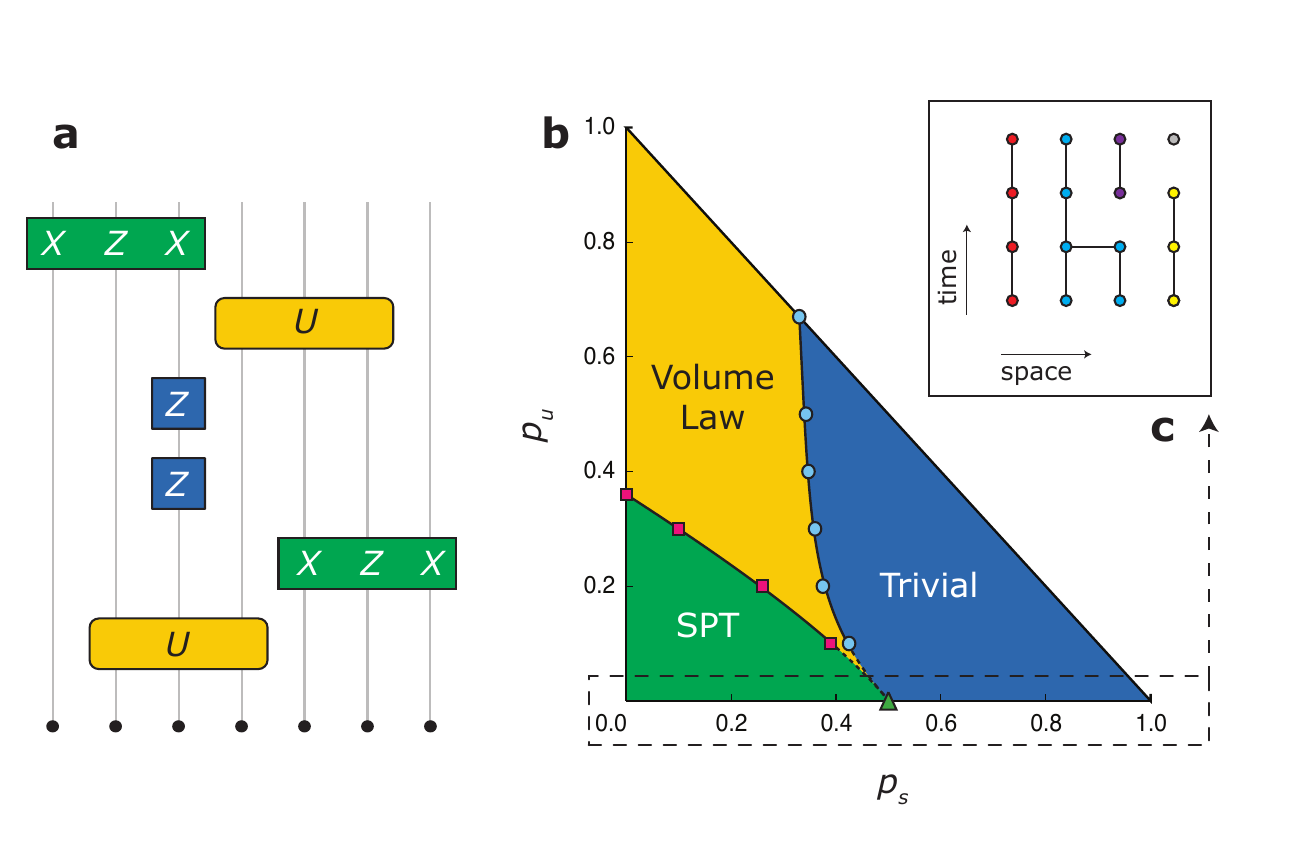}
  \caption{Schematic of the circuit and its corresponding phase diagram. \textbf{a}, Schematic diagram of a typical quantum circuit. Yellow (light) boxes corresponds to a three qubit random Clifford unitary, blue and green boxes represent projective measurements. \textbf{b}, The phase diagram describing the entanglement structure of the steady state. Red squares and blue circles are obtained from numerical simulations, while the rest of the phase boundaries are extrapolated. \textbf{c}, Mapping the dynamics of the random circuit on the $p_u=0$ axis to the $2$D percolation on a square lattice.}
  \label{fig_main}
\end{figure}

In this work, we consider the competition between these two types of measurements with each other as well as with the unitary dynamics. This raises the question of whether the notion of a topological phase is well-defined in random quantum circuits that include both unitary dynamics and local measurements.
To make progress in answering this question, we consider a (1+1)D quantum circuit model comprised of three
 elements: (a) Measurement of stabilizer operators that stabilize a $\mathbb{Z}_2\times\mathbb{Z}_2$ SPT realized by the ``cluster model"~\cite{raussendorf2001one,zeng2019quantum}; (b) Single-qubit measurements in the computational basis;
(c) Random, symmetry-allowed Clifford unitary gatesimulation possible~\cite{gottesman1997stabilizer})
. At each step of the circuit, one
element is selected at random with probability $p_t$, $p_s$, $p_u$ respectively ($p_t+p_s+p_u=1$) and applied at a random position in space. A typical snapshot of the circuit is shown in Fig. \ref{fig_main}a.

Using suitably defined order parameters, we discover a rich phase diagram, shown in Fig. \ref{fig_main}b. We find not only a stable SPT phase in an extended region of the phase diagram, but our results indicate a tricritical point, with logarithmic scaling of entanglement entropy, separating the volume law, trivial and SPT phases \textit{in the absence of unitary dynamics} $p_u=0$, i.e. when only measurements are present. The existence of this tricritical point implies that a volume-law phase can be stabilized by an infinitesimally small rate of unitary dynamics.

Moreover, we find an \textit{exact} analytical mapping that maps the case without unitary dynamics $p_u=0$
to two copies of a (non-standard) classical 2D percolation problem. Away from the $p_u= 0$ line, we extensively study the phase transitions numerically.
The numerical results are consistent with the correlation length critical exponent $\nu$ remaining the same on the phase boundaries all the way down to the tricritical point, which has $\nu=4/3$ based on the analytical mapping to percolation. On the other hand, we find that the coefficient of the logarithmic scaling of the entanglement entropy changes significantly, suggesting that the CFT description changes along the phase boundaries.

\section{Model}

We study a family of (1+1)D random quantum circuits that realize quantum trajectories extrapolating between wave functions in an SPT phase, a trivial product state, and a volume-law entangled phase.

We take our SPT to be the $\mathbb{Z}_2\times \mathbb{Z}_2$ symmetry protected phase realized by the cluster model defined on an open chain of $N$ qubits (we take $N$ even throughout) in (1+1)D~\cite{raussendorf2001one,zeng2019quantum},
\begin{equation}\label{eq_cluster_ham}
  H_0=- \sum_{i=2}^{N-1} X_{i-1}Z_iX_{i+1},
\end{equation}
where $X_i$ and $Z_i$ denote Pauli matrices.
Note that all terms commute with each other and therefore this model is exactly solvable.
This model realizes a SPT phase \cite{son2011quantum,santos2015rokhsar,tsui2017phase} protected by the $\mathbb{Z}_2\times \mathbb{Z}_2$ symmetry generated by
\begin{align}
  G_1=\prod_{i \text{ is even}}Z_i ~;~
  G_2=\prod_{i \text{ is odd}}Z_i.\label{eq_symmetry_generator}
\end{align}
 We say an eigenstate of $H_0$ is a symmetry invariant eigenstate if it is an eigenstate of all terms in $H_0$ as well as $G_1$ and $G_2$.
 All symmetry invariant eigenstates within the same symmetry sector can be related to each other by a symmetry-preserving constant depth local unitary circuit.

\begin{figure}
  \includegraphics[width=\columnwidth]{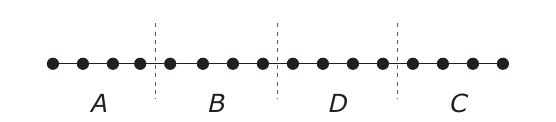}
  \caption{The $1$D chain cutting used to define the generalized topological entanglement entropy.}
  \label{fig:cutting}
\end{figure}

On an open chain, a particular generalization of entanglement entropy~\cite{zeng2015gapped,zeng2016topological,zeng2019quantum,fromholz2020entanglement} can be used as an order parameter for this SPT phase. Consider dividing the system as shown in Fig.~\ref{fig:cutting}.
The generalized topological entanglement entropy $S_\text{topo}$ is defined as
\begin{equation}\label{eq:s_topo_defenition}
  S_\text{topo} \equiv S_{AB}+S_{BC}-S_B-S_{ABC}.
\end{equation}
$S_{AB}$ stands for the von Neumann entanglement entropy of the region $A\cup B$ in the chain. Other terms are defined similarly. One can show that for all symmetry invariant eigenstates of $H_0$, $S_{\text{topo}} =2$.

To realize a wave function in this SPT phase, that is, a symmetry invariant eigenstate of $H_0$, we can for example use a quantum circuit that starts with an arbitrary eigenstate of $G_1$ and $G_2$ and then proceed to measure all stabilizer operators $g_i\equiv X_{i-1}Z_iX_{i+1}$.

To realize wave functions in the trivial phase, we use a quantum circuit that measures all single qubit operators in the $Z_i$ basis. The choice of the single qubit measurement basis $Z_i$ is fixed by demanding all measurement operators commute with the symmetry generators $G_1$ and $G_2$ (see Supplementary Section 1 for the case with symmetry violating measurements). All wave functions in the trivial phase have $S_\text{topo}=0$.

To realize wave functions in the volume law phase, we use random Clifford unitary gates that are allowed by the symmetry. The simplest class of gates to consider would be two qubit nearest-neighbor random unitaries. However, due to the symmetry restrictions, this set is not effective in entangling the qubits.
Ergo, we work with three-qubit random unitary gates.

We are now in a position to construct our full quantum circuit model: We start with the $\ket{0}^{\otimes N}$ state.
In each updating step we either: (a) apply a random 3-qubit Clifford unitary between qubits $i-1$, $i$ and $i+1$ with probability $p_u$, for a random $i$ drawn from $2,\cdots,N-1$, (b) measure the single qubit operator $Z_i$ with probability $p_s$, for a random $i$ drawn from $1,\cdots,N$, or (c) measure the stabilizer $g_i\equiv X_{i-1}Z_iX_{i+1}$ with probability $p_t=1-p_s-p_u$, for a random $i$ drawn from $2,\cdots,N-1$. A time step is defined as $N$ consecutive updating steps.

In the limiting case $p_u=1$ and $p_s=0$, the random unitary circuit drives the system into a volume law phase, whereas for the other two limiting cases, i.e. $p_u=0, p_s=0$ and $p_u=0,p_s=1$, the system is in an area law phase, one with SPT order and the other without.

We detect the presence of the different phases in several distinct ways. First, at each time step we calculate $S_{\text{topo}}$, averaged over quantum trajectories, and run the circuit until a steady state value is obtained. In addition to $S_{\text{topo}}$, to detect the phase transition from the area to volume law phase we extensively use the order parameter originally introduced in Ref.~\cite{gullans2019scalable}. To do so, first we run the circuit for time $2N$ to reach the steady state. Then, we entangle an ancilla qubit to the two qubits in the middle of the chain by measuring the following stabilizers,
\begin{equation}\label{eq_anc_stb}
  Z_{N/2-1} Z_a,\quad Z_{N/2+1} Z_a,\quad X_{N/2-1} X_a X_{N/2+1},
\end{equation}
where $X_a$ and $Z_a$ act on the ancilla qubit. Note that all three stabilizers commute with the symmetry generators $G_1$ and $G_2$. Next, we let the circuit run for an extra $\mathcal{O}(N)$ time steps, and then measure the entanglement entropy of the ancilla qubit. As shown in Ref.~\cite{gullans2019scalable}, if the system is in the area law phase, the ancillla's entanglement entropy $S_a$ should be zero by the time we measure it while in a volume law phase, the ancilla should be still entangled with the system.

We also use a slightly modified version of the ancilla order parameter\cite{gullans2019scalable}, which we call the scrambled ancilla order parameter denoted by $\tilde{S}_a$, such that instead of $1$ ancilla we use $10$ and instead of measuring the stabilizers listed in equation \eqref{eq_anc_stb} the ancillas are entangled to the system via $10$ time steps of a scrambling circuit, where at each updating step a random (non-symmetric) 3-qubit Clifford gate is applied to three randomly drawn qubits. As was the case for $S_a$, we measure the entropy of the ancilla subsystem after the qubit chain evolves $O(N)$ time steps under the symmetric random circuit. While in the trivial phase the ancilla subsystem would have been entirely disentangled from the qubit chain, giving $\tilde{S}_a=0$, in the SPT phase the ancilla subsystem should have remained entangled to the two edge degrees of freedom which are protected by the symmetry, resulting in $\tilde{S}_a=2$. In the volume law phase the ancilla subsytem should remain entangled to the bulk as well and hence $\tilde{S}_a>2$.

It turns out that compared to $\tilde{S}_a$ and $S_{topo}$, $S_a$ shows a sharper SPT to volume law phase transition when $p_u>0$ (see Supplementary Section 11) --and hence it is used to extract the corresponding critical exponents-- but is unable to detect the topological phase transition at $p_u=0$. On the other hand, $\tilde{S_a}$ can be used as an experimentally accessible probe to detect the phase transition at $p_u=0$.

We note that a type of Edwards-Anderson glass order parameter can also be used to detect the topological phase (see Supplementary Section 4), although it cannot distinguish the trivial and volume law phases.

Finally, we note that
the random quantum circuits studied here, viewed as a quantum channel, eventually transform the initial state of the system into the maximally mixed state allowed by the symmetry (see Supplementary Section 3 for a proof and a bound on how fast this happens). Therefore, the steady state expectation value of \textit{any} operator stays the same throughout the phase diagram and thus cannot serve as an order parameter.

\section{Mapping The Case Without Unitary Dynamics $p_u=0$ to Classical Percolation}

Here we show how to map the entire $p_u=0$ line in the random circuit presented above to two copies of a classical $2D$ percolation problem on a square lattice.
This percolation model is non-standard, although our numerical results indicate that it has the same critical properties as the standard classical percolation model on the square lattice. There is a distinct but closely related random quantum circuit that we define in Supplementary Section 9 which does map directly to (two copies of) standard classical percolation.

Let us divide the operators measured by the random circuit into two sets. One set, which we call the odd site operators, is comprised of single qubit operators $Z_i$ for odd $i$ alongside the stabilizers $g_j$ which \textit{end} on the odd sites, i.e. for even $j$.  The even site operators are defined analogously. Note that each member of one set commutes with all elements of the other set.

Let us focus on the measurements of odd site operators. Consider the $N/2~\times~M$ square lattice as shown in Fig. \ref{fig_main}c, where $M$ is the total number of updating steps in the circuit. We call this lattice the odd sites' percolation lattice.  The $N/2$ vertices on each row corresponds to the odd sites of the system and we label them accordingly. The vertical (horizontal) links ending (residing) on the $m$'th row are related to the $Z_i$ ($g_j$) measurements in the $m$'th step of the circuit in the following way: if $Z_i$ \textit{is not} measured at updating step $m$, we draw a vertical link between the $(i,m-1)$ and $(i,m)$ vertices. Also if the stabilizer $g_j$ \textit{is} measured at step $m$, we draw a horizontal link between the $(j-1,m)$ and $(j+1,m)$ vertices. At the end, we assign a unique color to each connected cluster of vertices. We construct the even sites' percolation lattice analogously. The randomness of the quantum circuit translates into random connections in the percolation lattices; the probability distributions for the links in the percolation lattice are detailed in Supplementary Section 5.

The entanglement structure of the system at step $M$ can be extracted from the colors of the vertices on the last row of the two aforementioned percolation lattices. As the following proposition makes precise, qubits of the same color make up their own SPT state:
\begin{prop}\label{prp_coloring}
  Group the qubits based on their color on the last row of the percolation lattice. Let $A^j=\{q_i\}_{i=1}^n$ denote the ordered set of qubit indices corresponding to $j$'th color; that is, the $q_i$ label a set of qubits all with the same color at step $M$. Then, up to a minus sign, the operators that stabilize the state of the system at step $M$ are of the following form,
  \begin{equation}
     \prod_{i=1}^n Z_{q_i}\quad \text{and}\quad g_{q_i,q_{i+1}}\quad \text{for }i=1,2,\cdots,n-1,
  \end{equation}
  where $g_{q_i,q_{i+1}}$ is defined as
  \begin{equation}
    g_{i,j}=X_{i}~\qty[\prod_{k=0}^{\frac{j-i}{2}-1}Z_{i+2k+1}]~X_{j}.
  \end{equation}
\end{prop}
By considering similarly defined stabilizer operators for all different colors ($A^j$'s with different $j$), we get a complete set of stabilizers that specify the state of the system.
The proof of Proposition \ref{prp_coloring} is left for Supplementary Section 8.

As shown in Lemma 1 in Supplementary Section 2, the minus sign ambiguity in Proposition \ref{prp_coloring} has no bearing on the entanglement spectrum of the system's state. Thus the percolation lattices exactly determine the (von Neumann or R\`enyi) entanglement entropy for any subset of qubits.

\section{Numerical Results}

We start by briefly reviewing the quantities we numerically calculate to obtain the phase diagram and to characterize the critical phase boundaries.

A signature of criticality in (1+1)D systems is the logarithmic scaling of the entanglement entropy. Thus, we calculate the entanglement entropy at the $t$th time step (which corresponds to $tN$ updating steps), $S(x,L;t)$ of a subsystem of length $x$ for a system of total length $L = N$, averaged over all of the quantum trajectories of the circuit.

In the large time limit, this averaged entanglement entropy saturates to a logarithmic form at the phase transitions as in (1+1)D CFTs~\cite{calabrese2009entanglement}:
\begin{equation}\label{eq_ee_log}
  S(x,L)=a_x \log(\frac{L}{\pi}\sin \frac{\pi x}{L})+b.
\end{equation}
We can also characterize the entanglement growth with time. At criticality, for timescales much smaller than the saturation time we have,
\begin{equation}\label{eq_dynamical_z_in_s_t}
  S(x,L;t)={a_t} \log(t)+b'.
\end{equation}
Note that as opposed to unitary CFTs the coefficient of the logarithmic scaling $a_x$ is not given by the central charge of any underlying CFT. In the context of the area law to volume law transition, Ref.~\cite{fisherludwig20} provides an appealing interpretation of $a_x$ and $a_t$ as universal quantities given by the scaling dimension of certain ``boundary condition changing" operators. $b$ and $b'$ are non-universal constants

Throughout the phase boundaries, we find $a_x=a_t$ within the margin of error, which is consistent with a dynamical exponent $z=1$, as the entanglement growth rate is similar along time and space directions.

We can use the averaged topological entanglement entropy, $S_\text{\text{topo}}$ as the order parameter to distinguish the three different phases: $S_\text{topo}$ would be extensive in the volume law phase, while in the thermodynamic limit it should converge to values $2$ and $0$ in the topological and trivial phases respectively. Let $S_{\text{topo}}(p,L)$ denote the steady state value of $S_\text{topo}$ when some tuning parameter (e.g. single qubit measurement probability) is $p$ and system size is $L$. On general grounds, we expect the following scaling form in the vicinity of the  critical point,
\begin{equation}\label{eq_finite_size_scaling}
  S_{topo}(p,L)=F((p-p_c)L^{1/\nu}),
\end{equation}
where $F(x)$ is some unknown function, $p_c$ is the critical value of tuning parameter $p$, and $\nu$ is the correlation length critical exponent,
$\xi\propto|p-p_c|^{-\nu}$.

As explained in the Model section, the entanglement entropy of a suitably entangled ancilla system, $S_a$ or $\tilde{S}_a$ can also be used as the order parameter to distinguish the volume law phase from the other two area law phases. Assuming the dynamical exponent $z=1$, for the ancilla entropy $S_{a}$ we have~\cite{gullans2019scalable},
\begin{equation}
  S_{a}(p,L,t)=G((p-p_c)L^{1/\nu},t/L),
\end{equation}
where $G(x)$ is some unknown function. $\tilde{S}_a$ has a similar scaling form.

\begin{figure}
  \includegraphics[width=\columnwidth]{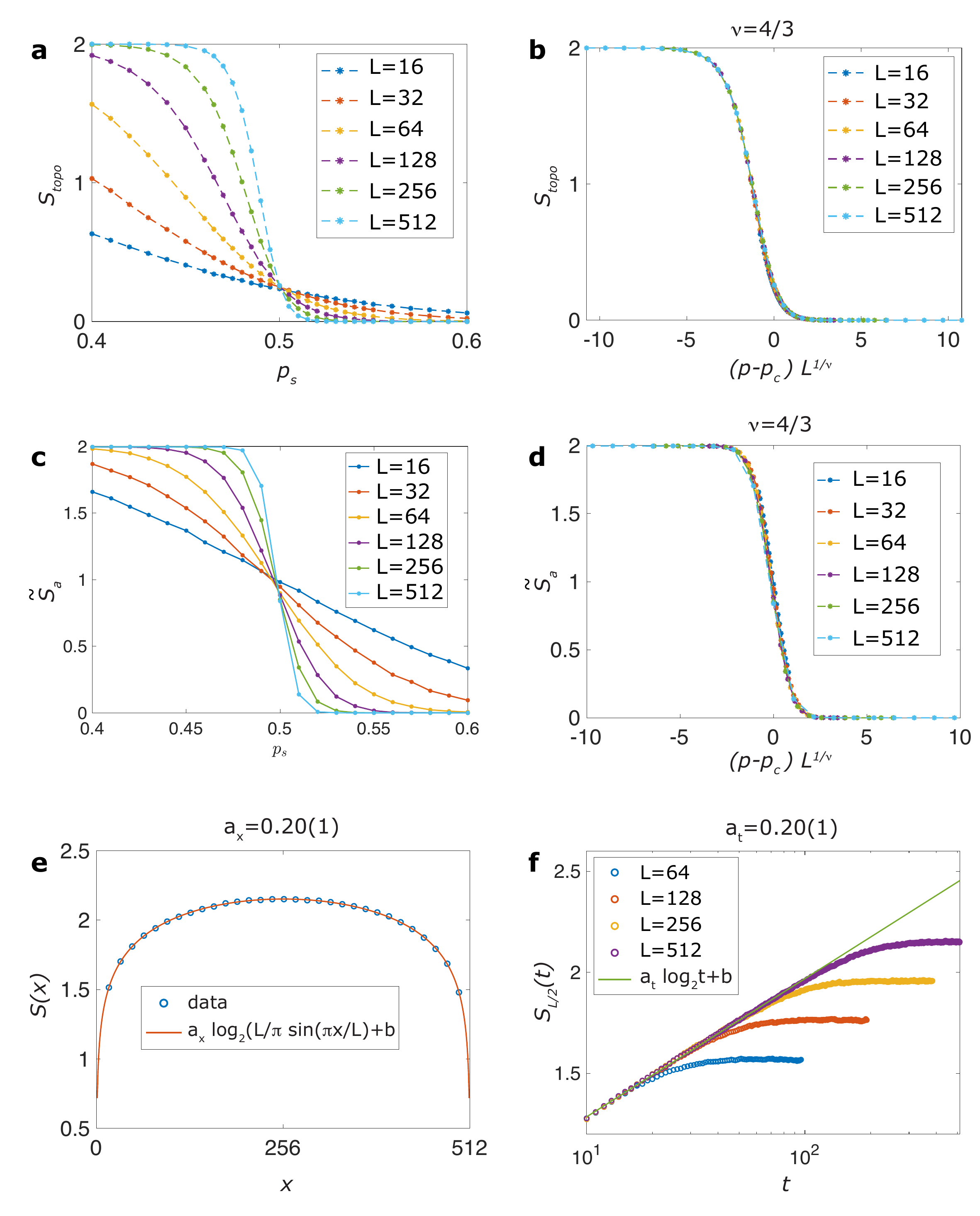}
  \caption{The numerical results pertaining to the tricritical point at $p_u=0$. \textbf{a}, Topological entanglement entropy $S_\text{topo}$ near the tricritical point versus single qubit measurement probability $p_s$. \textbf{b}, Scaling collapse of the data in panel a. \textbf{c}, Ancilla entropy $\tilde{S}_a$ measured $t=N$ time steps after scrambling. \textbf{d}, Scaling collapse of the data in panel c. \textbf{e}, The entanglement entropy of the $[0,x]$ segment of the chain, $S(x,L)$, at late times for $p=p_c$ and $L=512$. \textbf{f}, The entanglement entropy of the half-chain versus time for $p_s=p_c$. All entropies are in units of $\log 2$. See Supplementary Section 8 for an analytical derivation of the $a$ coefficient using the percolation map.}
  \label{fig_u0plots}
\end{figure}

We now present our numerical results.  We study system sizes up to 512 qubits and average over $10^5$ random quantum trajectories. We start with the $\ket{0}^{\otimes N}$ state and let the circuit run for $2N$ time steps for the system to reach the steady state. We have explicitly verified that saturation is reached before $t=2N$. After entangling the ancilla qubit, we simulate the system for an additional $\mathcal{O}(N)$ time steps to calculate $S_{a}$ (as explained above).

 Fig. \ref{fig_u0plots} shows numerical results along the $p_u=0$ line. Fig. \ref{fig_u0plots}a and c show the steady state value of $S_\text{topo}$ and $\tilde{S}_a$ versus $p_s$ for different system sizes. As is evident from both diagrams, there is a clear continuous phase transition at $p_c=1/2$ in the thermodynamic limit. A simple argument based on duality shows that if there is a continuous phase transition between the trivial and topological phase, it has to be at $p_s =1/2$. This duality argument is provided in Supplementary Section 6. Interestingly we find that $S_a$ seems to be unable to capture the area-law to area-law phase transition at $p_u=0$, at least for numerically accessible systems sizes. On the other hand,  From collapsing the data near the critical point $p_c=1/2$, we find $\nu=4/3$ results in a near perfect collapse (see Fig. \ref{fig_u0plots}c and d).

Fig.~\ref{fig_u0plots}e shows the steady state value of entanglement entropy $S(x)$ of the subregion $[1,x]$ at the critical point $p_u=0$ and $p_s=1/2$, for $L=512$. As shown, the entanglement entropy fits the CFT form of equation  \eqref{eq_ee_log} with $a_x=0.20(1)$.

Fig. \ref{fig_u0plots}f shows the entanglement entropy of the half chain versus time at $p=p_c$ for different chain sizes. The entanglement entropy grows logarithmically with time, until the finite size effects show up. By comparing the corresponding fitted analytical expressions we find $a_t=a_x=0.20(1)$.

We note that the values of $\nu$, $z$, $a_x$, and $a_t$ at the $p_u = 0$, $p_s = 1/2$ transition agree with the values of our other random measurement-based quantum circuit model presented in Supplementary Section 9, which in turn maps to (two copies of) the standard classical link percolation problem on the square lattice. Our results are thus consistent with the $p_u = 0$, $p_s = 1/2$ transition studied in Fig. \ref{fig_u0plots} being governed by (two copies of) the standard classical percolation fixed point.

\begin{figure*}
  \includegraphics[width=\textwidth]{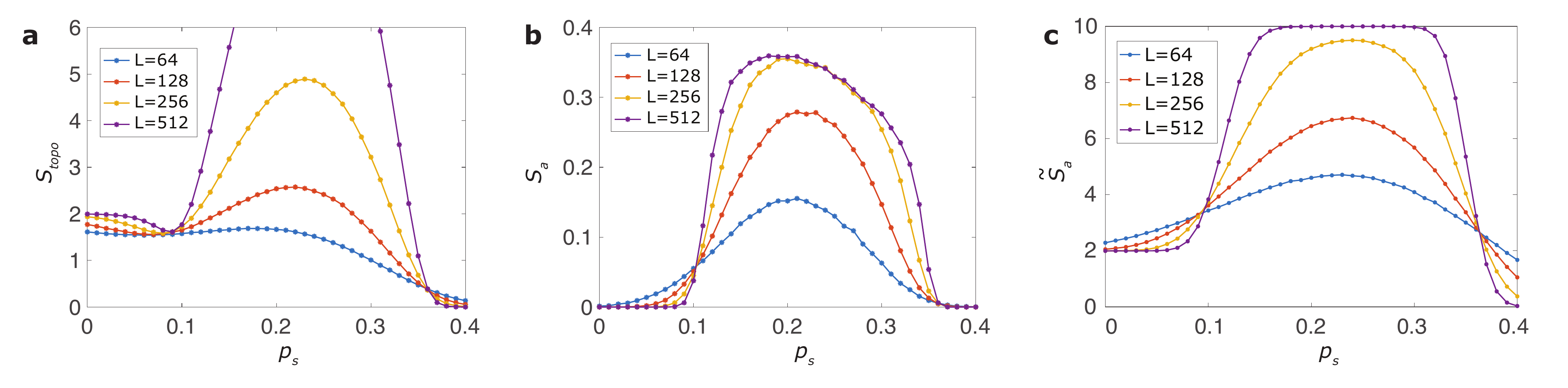}
  \caption{The phase transitions across the $p_u=0.3$ line. \textbf{a}, $S_\text{topo}$ verses $p_s$. \textbf{b}, The ancilla entropy $S_a$ measured $t=N$ time steps after it was entangled, versus $p_s$. \textbf{c}, The ancilla entropy $S_a$ measured $t=N$ time steps after it was entangled, versus $p_s$. In all three panels, the first crossing corresponds to the phase transition from the SPT phase into the volume law phase while the second crossing corresponds to the phase transition from the volume law phase to the trivial phase. The critical points are marked on the phase diagram in Fig.\ref{fig_main}b as well.}
  \label{fig_s_t_s_a}
\end{figure*}

We now proceed to the case with unitary dynamics $p_u\neq0$. Fig. \ref{fig_s_t_s_a} shows $S_\text{topo}$, $S_a$ and $\tilde{S}_a$ versus $p_s$ for the fixed value of $p_u=0.3$. For $p_s=0$, the system is in the topological phase as can be seen from Fig. \ref{fig_s_t_s_a}a. By increasing $p_s$, the entropies exhibit a continuous phase transition to the volume law phase at first and then another continuous phase transition to the trivial phase.

By using analogous plots for different values of $p_u$, we can determine the $2$D phase diagram in the $(p_s,p_u)$ space. The result is illustrated in Fig. \ref{fig_main}b. Note that since the probability of measuring a stabilizer is $1-p_u-p_s$, the phase diagram is restricted to the region $p_u+p_s \le 1$. The data points on the plot have been extracted using numerical simulations and then the schematic phase diagram is drawn based on them. For more detailed results used in obtaining the phase diagram see Supplementary Section 11.

The SPT/volume law phase boundary intersects the $p_u$ axis at $p_u=0.355(3)$ and the volume law/trivial phase boundary ends at $p_u=0.663(4)$ on the $p_u+p_s=1$ line. Our numerical simulations demonstrate that the volume law phase still exists for $p_u$ as low as $0.1$. Unfortunately, clearly detecting the SPT to volume law transition requires increasingly large system sizes as $p_u$ is lowered (See Supplementary Section 10 for details). Therefore, we extrapolate the phase diagram for smaller values of $p_u$. By following the trend of the data points, it appears that the volume law phase survives all the way down to $p_u=0$, hence suggesting that the critical point at $p_s=0.5$ and $p_u=0$ is actually a tricritical point. This in turn means that at $p_s = 1/2$, arbitrarily sparse random Clifford gates in the quantum circuit can still drive the system into the volume law phase.

By using the scaling form in equation  \eqref{eq_finite_size_scaling} and collapsing the data, we can extract the correlation length critical exponent $\nu$ along the phase boundaries. Taking into account the margins of error, our numerical results are consistent with $\nu = 4/3$ everywhere along the phase boundaries (see Supplementary Section 11 for the corresponding plots and numerical values). However $a_x = a_t$ changes significantly along the phase boundaries at the largest system sizes we have studied. If the $a_x = a_t$ that we
extract are indeed close to their values in the thermodynamic limit, this suggests that the volume to area law critical lines may be related to two copies of the classical percolation fixed point by marginal deformations.

%

Entanglement phase transitions involving topological or SPT phases, also seem to be closely related to quantum error correction. In particular, the rapid stabilizer measurements are reminiscent of syndrome measurements in active error correction schemes. Moreover, random single qubit measurements can be viewed as faulty syndrome measurements or qubit decoherence, while unitary dynamics models the random noise affecting the qubits. In this context, ``entanglement phase transitions" could be related to ``error-thresholds" beyond which the long range entanglement structure of the code space, which is responsible for the topological protection of the encoded information, is entirely lost, hence rendering recovery of logical information impossible. Within this framework, our results might have natural applications to quantum error correcting codes.
Note that this is a different analogy to quantum error correction than the one presented in Ref.~\cite{altman19error,gullan19purify}, where the volume law phase is considered to be a quantum error correcting code.

\section{Acknowledgements}

We thank M. Hafezi, H. Dehghani, and A. Nahum for helpful comments. We are especially grateful to M. Gullans and D. Huse for suggesting the modified ancilla order parameter and discussions regarding its saturation value in the topological phase. The authors acknowledge the University of Maryland supercomputing resources (http://hpcc.umd.edu) made available for conducting the research reported in this paper. A.L and M.B are supported by NSF CAREER (DMR- 1753240), Alfred P. Sloan Research Fellowship, and JQI- PFC-UMD. Y.A is supported by National Science Foundation NSF DMR1555135 and JQI-NSF-PFC.

\bibstyle{naturemag}
\bibliography{mybib}

\pagebreak
\clearpage

  \begin{center}
  \textbf{\large Supplementary Materials}
  \end{center}

\setcounter{equation}{0}
\setcounter{figure}{0}
\setcounter{page}{1}
\setcounter{section}{0}
\renewcommand{\theequation}{S\arabic{equation}}
\renewcommand{\thefigure}{S\arabic{figure}}
\renewcommand{\thesection}{\arabic{section}}

\section{Methods}\label{methods}
\subsection{Binary Representation of the Stabillizer Circuit}

We use the binary representation of the stabilizer formalism to simulate the Clifford circuits studied in this work. This representation is based on the observation that, up to some phase factor, any Pauli string operator $s$ over $N$ qubits can be uniquely mapped to a binary vector $ w=( u, v) \in \mathbb{Z}_2^{2N}$ where $ u,  v \in \mathbb{Z}_2^N$ and
\begin{equation}
  s=e^{i\theta}~\prod_{i=1}^N X_i^{u_i}~\prod_{i=1}^N Z_i^{v_i}
\end{equation}
If Pauli string operators $s_1$ and $s_2$ correspond to vectors ${w}_1$ and ${w}_2$, their multiplication $s_1 s_2$ corresponds to ${w}_1+{w}_2$. Moreover, $[s_1,s_2]=0$ if and only if ${w_1}^T  g {w}_2=0$, where $ g$ is the $2N \times 2N$ matrix defined as
\begin{equation}\label{eq_inner_product}
 g =\left(
\begin{array}{c|c}
~0_{N\times N}~& \mathbbm{1}_{N\times N} \\ \hline
\mathbbm{1}_{N\times N} & 0_{N\times N} \\
\end{array}
\right)
.
\end{equation}
It is also easy to apply Clifford unitaries in the binary representation. Let $U$ be a unitary in the Clifford group. Since $U$ belongs to the Clifford group, the images of $X_i$ and $Z_i$ under $U$, i.e. $U X_i U^\dagger$ and $U Z_i U^\dagger$ are themselves Pauli string operators  and have binary representations in $\mathbb{Z}_2^{2N}$. Let $M_U$ be the $2N \times 2N$ matrix whose first and second $N$ columns corresponds to the images of $X_i$s and $Z_i$s under $U$ respectively, for $i=1,\cdots,N$. It is easy to see that, if $ w$ is the binary representation of a Pauli string $s$, the binary representation of $U s U^\dagger$ would be given by the matrix multiplication $M_U  w$ in $\mathbb{Z}_2$.

Given a stabilizer set $\mathcal{S}$,  a $N \times 2N$ stabilizer matrix $M_{\mathcal{S}}$ can be formed  by taking the binary representation of the elements of $\mathcal{S}$ as its rows. For example, the stabilizer matrix which corresponds to the state $\ket{0}^{\otimes N}$ is given as
\begin{equation}
  M_\mathcal{S}=(0_{N\times N} | \mathbbm{1}_{N \times N})
\end{equation}
 One may keep track of the phase factors using an additional $N$ element vector, but since we are interested in the entanglement structure which is independent of the phase factors (see Lemma.\ref{lm_indistinguishable}), we ignore the phase factor in what follows.

In our numerics we use the stabilizer matrix of the system to keep track of the entanglement dynamics of the system. As was discussed above, applying a unitary $U$ would transform $M_\mathcal{S}$ to $M_\mathcal{S}M_U^T$ where $M_U$ is the binary representation of $U$ and $T$ stands for transpose (note that the stabilizers are stored as the rows of $M_\mathcal{S}$ rather than its columns). It is straightforward to keep track of the Clifford measurements as well, due to the Gottesman-Knill theorem\cite{gottesman1998heisenberg,nielsen2002quantum}.
 Let $s_\ast $ represent the Pauli string operator that is being measured. First we find the stabilizers in $\mathcal{S}$ that do not commute with $s_m$, which can be done efficiently by computing $M_\mathcal{S} g s_\ast$ with $g$ as defined in equation \eqref{eq_inner_product}. If $s_\ast$ commutes with every stabilizers in $\mathcal{S}$ then measuring it has no effect on the state of the system. On the other hand, if $s_\ast$ doesn't commute with some stabilizers in $\mathcal{S}$, say $s_1,\cdots,s_m $,
the stabilizer set of the system after measuring $s_\ast$ can be obtained by replacing $s_1$ with $\pm s_\ast$ and $s_i$ with $s_1 s_i $ for $i=2,\cdots,m$, where the $\pm$ sign is chosen at random. Since we are ignoring the phase factors in the binary representation, this amounts to replacing the row corresponding to $s_1$ with the binary representation of $s_\ast$ and adding the binary representation of $s_1$ to the rows corresponding to $s_2\cdots s_m$.

To sample the three-qubit symmetric Clifford unitary set, we use the procedure outlined in Ref.\cite{li2019measurement} to generate all possible binary representations of three qubit Clifford unitaries and then choose the symmetric subset by explicitly checking whether a unitary respects the $\mathbb{Z}_2 \times \mathbb{Z}_2$ symmetry.

As the last remark in this section, we note that given the stabilizer matrix $M_\mathcal{S}$, the entanglement entropy of a subset  $A$ of the qubits can obtained via\cite{nahum2017quantum}:
\begin{equation}
  S_A=\text{rank}(M_\mathcal{S}\big|_A)-n_A,
\end{equation}
where $M_\mathcal{S}\big|_A$ is the submatrix of $M_\mathcal{S}$ obtained via keeping only the columns which correspond to the qubits in $A$, $n_A$ is the number of qubits in $A$ and the rank is computed in $\mathbb{Z}_2$.

\subsection{Estimating the Errors}

In this section we briefly summarize the procedure that was used to estimate the numerical values of  parameters and their corresponding errors.

The critical value $p_c$ can be found by plotting the order parameter for different system sizes and locating the scale invariant point at which all the curves for  different system sizes cross. The reported value of $p_c$ corresponds to the crossing point of the order parameter curves for $L=512$ and $L=256$, while the curves for smaller system sizes are used to estimate the error. We find $p_c(L)$ for $L=512,256$ and $128$, where $p_c(L)$ is defined as the crossing point between curves of system sizes $L$ and $L/2$. The y-intercept of the linear fit to $p_c(L)$ as a function of $1/L$ gives an estimate for $p_c(L\to\infty)$. The error in $p_c$ is then estimated by the difference between the extrapolated value $p_c(L \to \infty)$ and its value at $L=128$. We use $S_a$ as the order parameter to detect the phase transition from SPT to volume law entangled phase, while we use $S_\text{topo}$ for the phase transition from the volume law entangled phase to the trivial phase.  The reason for using two different order parameters is that, while $S_\text{topo}$ has less noise than the ancilla order parameter, one has to go to larger system sizes to properly detect the SPT to volume law phase transition using this order parameter. On the other hand, although $S_a$ has to be averaged over higher number of realizations, it displays a sharper crossing compared to $S_\text{topo}$ or $\tilde{S}_a$ (See Supplementary Section 11).

The value of correlation length critical exponent $\nu$ is found from the data collapse. Assume a certain quantity, say $S_a$, has the following finite size scaling form:
\begin{equation}
  S_a(p,L)=F((p-p_c)L^{1/\nu})
\end{equation}
for some arbitrary function $F$. It follows that if one plots $S_a$ as a function of $(p-p_c)L^{1/\nu}$, for the right choice $\nu$, all the data points would collapse on the $y=F(x)$ curve. To find the best collapse, we use the objective function $\epsilon(\nu)$ defined as:
\begin{align}
  \epsilon(\nu)=\frac{1}{n-2}\sum_{i=2}^{n-1} \qty(y_i-\bar y_i)^2,
\end{align}
  where,
  \begin{align}
    \bar y_i&= \frac{(x_{i+1}-x_i)y_{i-1}-(x_{i-1}-x_i)y_{i+1}}{x_{i+1}-x_{i-1}},
  \end{align}
  with  $x_i=(p_i-p_c)L_i^{1/\nu}$ and $y_i=S_a(p_i,L_i)$. Here $i$ labels the $i$'th data point (including system sizes $L=64,128,256$ and $512$) sorted based on their $x$ value such that $x_1<x_2<\cdots<x_n$, and $n$ denotes the total number of data points. $\bar y_i$ is the expected value of $y_i$ if it was on the line passing through the two adjacent data points. For the perfect collapse and in the limit of infinitely close data points, $\epsilon(\nu)$ would vanish. To obtain the best collapse, we find the value $\nu^\ast$ which
   minimizes the objective function $\epsilon(\nu)$ for a given set of numerical data. To estimate the error, we find the $\nu$ interval for which  $\epsilon(\nu)<2~\epsilon(\nu^\ast)$.

   The numerical values of $a_t$ and $a_x$ (See equation \eqref{eq_ee_log} and equation \eqref{eq_dynamical_z_in_s_t}) are obtained by fitting the numerical data for $L=512$ to the analytical expressions via the method of least squares. Because these equations are field theory results and are valid only in length scales much larger than lattice spacing, we exclude data points corresponding to first and last $10$ sites before fitting the data to equation \eqref{eq_ee_log}. As for equation \eqref{eq_dynamical_z_in_s_t}, we exclude data points for $t<10$ as well as data points close to the saturation value of half chain entropy, $S_{L/2}(t)>0.9~S_{L/2}(\infty)$.
   Let $\varepsilon_1$ denote the error that characterizes the fit quality. There is also a systematic source of error related to finite size effects, which we denote by $\varepsilon_2$. To estimate $\varepsilon_2$, we evaluate $a_x(L)$ and $a_t(L)$ for $L=128,256,512$, find the $y$-intercept of the linear fit to  $a_x(L)$ (and $a_t(L)$) as a function of $1/L$ and then estimate $\varepsilon_2$ as the difference between the $y$-intercept and the parameters evaluated at $L=128$. The reported error is then $\max(\varepsilon_1,\varepsilon_2)$.

\section{Circuits which do not respect the symmetry}\label{apx_symmetry_violating}

To verify that the stability of the SPT phase in the circuit model of the main text depends on the circuit respecting the symmetry, we consider the same circuit model at $p_u=0$ but with $Z_i$ measurements replaced by $X_i$ measurements. Fig.\ref{fig_X_circuit} shows the topological entanglement entropy in the steady state versus $p_s$ for the aforementioned modified circuit. As expected, any infinitesimal $p_s$ will destroy the topological phase in the thermodynamic limit.

\begin{figure}[h]
  \includegraphics[width=.7 \columnwidth]{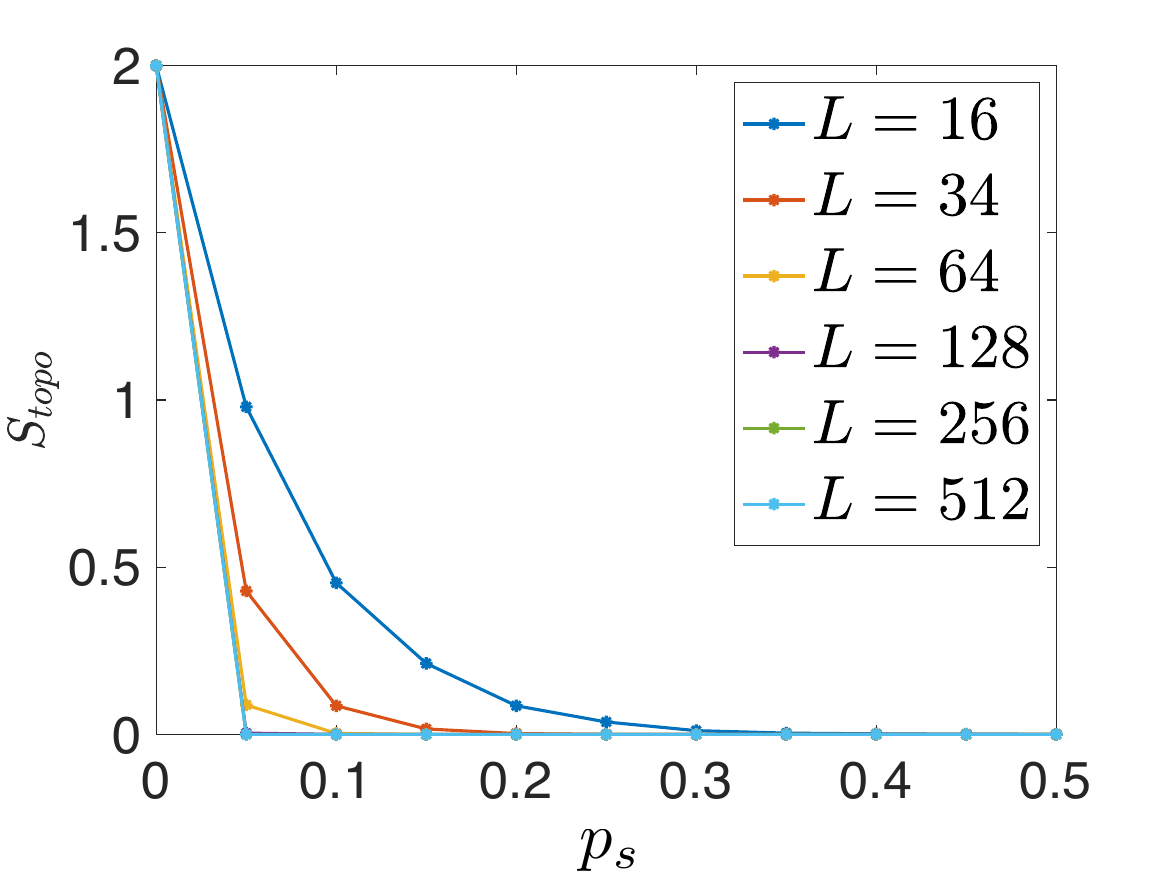}
  \caption{$S_\text{topo}$ versus $p_s$. The single qubit measurements are in the $X_i$ basis rather than the $Z_i$ basis. }
  \label{fig_X_circuit}
\end{figure}

\section{Entanglement entropy in stabilizer formalism}\label{apx_stabilizer_formalism}

 A stabilizer state $\ket{\psi}$ over $N$ qubits is specified by a set $\mathcal{S}=\{s_1,\cdots,s_N\}$ of $N$ independent, mutually commuting Pauli strings operators $s_i$, such that
 \begin{equation}
   s_i \ket{\psi}=\ket{\psi}.
 \end{equation}
 Clearly, there are many equivalent choices of $\mathcal{S}$ that result in the same stabilizer state $\ket{\psi}$. However, given a stabilizer set $\mathcal{S}$, the elements $s_i$ generate an abelian group $\mathcal{G}_\psi=\langle s_1,s_2,\cdots ,s_N\rangle$ under multiplication, which is determined uniquely by the stabilizer state $\ket{\psi}$.

 The density matrix of the system in the stabilizer state $\ket{\psi}$ can be written as \cite{fattal2004entanglement}
 \begin{equation}\label{eq_stab_density}
   \rho=\ketbra{\psi}{\psi}=\frac{1}{2^N}\sum_{g\in \mathcal{G}_\psi}g.
 \end{equation}
 Given a bipartition of the qubits into two sets $A$ and $B$, the reduced density matrix of $\rho$ over $A$ can be obtained by tracing equation  \eqref{eq_stab_density} over $B$, which yields
 \begin{equation}\label{eq_stab_red_density}
   \rho_A=\frac{1}{2^{n_A}}\sum_{g\in \mathcal{G}_{A,\psi} } g,
 \end{equation}
 where $n_A$ is the number of qubits in $A$ and $\mathcal{G}_{A,\psi}\subseteq \mathcal{G}_\psi$ is the subgroup of the stabilizers which are entirely contained in $A$, i.e. they act as identity on the qubits outside $A$. The von Neumann entanglement entropy of $\rho_A$ is given by
 \begin{equation}\label{eq_entropy_groupTheory}
   S_A(\ket{\psi})=n_A-\log_2|\mathcal{G}_{A,\psi}|,
 \end{equation}
 where $|\mathcal{G}|$ stands for the number of elements in group $\mathcal{G}$\cite{fattal2004entanglement}. Moreover, $R_\alpha(\rho_A)$, the Renyi entropy of order $\alpha$,  is actually independent of $\alpha$ and is equal to von Neumann entanglement entropy.

Let $\ket{\psi}$ be a stabilizer state specified by the stabilizer set $\mathcal{S}=\{s_1,s_2,\cdots,s_N\}$.
 Consider a closely related stabilizer state $\ket{\psi'}$ which is specified by the stabilizer set $\mathcal{S}'=\{(-1)^{n_1}s_1,(-1)^{n_2}s_2,\cdots,(-1)^{n_N}s_N \}$
 where each $n_i$ is either $0$ or $1$. The following Lemma shows that $\ket{\psi}$ and $\ket{\psi '}$ are indistinguishable as far as the entanglement entropy is concerned.
 \begin{lemma}\label{lm_indistinguishable}
 For $\ket{\psi}$ and $\ket{\psi'}$ defined as above and for any subset $A$ of the qubits,
 \begin{equation}
   S_A(\ket{\psi})=S_A(\ket{\psi'}).
 \end{equation}
 \end{lemma}
 \begin{proof}
   Let $\mathcal{G}\psi$ and $\mathcal{G}_{\psi'}$ denote the stabilizer groups associated with $\ket{\psi}$ and $\ket{\psi'}$ respectively. Consider the group homomorphism $h$ between $\mathcal{G}_\psi$ and $\mathcal{G}_{\psi'}$ defined by its action on the generators of $\mathcal{G}_\psi$ as
   \begin{align*}
     &h: \mathcal{G}_\psi \longmapsto \mathcal{G}_{\psi '}\\
     &h(s_i)=(-1)^{n_i} s_i.
   \end{align*}
   Since $h$ maps a generator set to another, it is bijective. Moreover, it is straightforward to verify that for any subset $A$ of qubits, $h$ maps $\mathcal{G}_{A,\psi}$ to $\mathcal{G}_{A,\psi'}$. The Lemma's claim then follows immediately from equation ~\eqref{eq_entropy_groupTheory}.
 \end{proof}

 Given a stabilizer state $\ket{\psi}$, one has the freedom to choose any $N$ independent elements from $\mathcal{G}_\psi$ to form the stabilizer set $\mathcal{S}$. We can use this gauge freedom to impose certain conditions on the elements of $\mathcal{S}$.

 Define the left (right) endpoint of a stabilizer $s$ to be the first (last) site on which $s$ acts non-trivially. Given a set of stabilizers $\mathcal{S}$, let $\rho_l(i)$ denote the the number of stabilizers whose left endpoint resides on site $i$ and define $\rho_r(i)$ similarly with regard to the right end points. As is shown in Ref. \cite{nahum2017quantum}, one can always choose $\mathcal{S}$ such that
 \begin{enumerate}
   \item For all sites we have $\rho_r(i)+\rho_l(i)=2$.
   \item If $\rho_l(i)=2$ (or $\rho_r(i)=2$) for a site $i$, the two corresponding stabilizers have a different Pauli operator at $i$.
 \end{enumerate}
 Such a stabilizer set $\mathcal{S}$ is said to be in the clipped gauge \cite{li2019measurement}. The utility of the clipped gauge is that the entanglement entropy has a simple form in this gauge. In particular, if the stabilizer set $\mathcal{S}$ is in the clipped gauge, the entanglement entropy of a contiguous region $A$ equals to half the number of stabilizers in $\mathcal{S}$ which have one endpoint in $A$ and another in its complement\cite{li2019measurement}.

 \section{Steady state density matrix}\label{apx_std_rho}

Let $\mathcal{Q}$ denote a specific realization of the quantum circuit laid out in the main text. If we fix the initial state to be $\ket{0}^{\otimes N}$ and run the same circuit many times, due to the randomness in the measurement outcomes, the final state of the system could be different each time. Instead of considering quantum trajectories, we can calculate the expectation value of operators over different runs by viewing the measurements in $\mathcal{Q}$ as quantum channels. Accordingly, the entire circuit can be described as a quantum channel $\mathcal{E}_\mathcal{Q}$, which transforms the  initial pure density matrix $\rho_0=\qty(\ketbra{0}{0})^{\otimes N}$ to a mixed final density matrix $\rho_\ast=\mathcal{E}_\mathcal{Q}(\rho_0)$.

$\rho_\ast$ can be used to compute the expectation values of measurements which are averaged over many runs of the circuit without post-selection on the measurement outcomes. In particular, if we run the circuit $\mathcal{Q}$ many times, measure a fixed operator $O$ each time and average the result over a large number of runs, the value we get would be
\begin{equation}
  \bar O =\tr(O \rho_\ast).
\end{equation}

Here we show that with probability one, $\rho_\ast$ is actually independent of the underlying circuit. In other words, if $\mathcal{Q}$ is any fixed quantum circuit chosen with the distribution associated with probabilities $0<p_s,p_t<1$ and $p_u<1$, the final density matrix of the system is always given by
\begin{equation}
  \rho_\ast=\mathcal{E}_\mathcal{Q}(\rho_0)=\frac{1}{2^{N-2}}\Pi_{G_1,+}\Pi_{G_2,+}.
\end{equation}
where $\Pi_{G_i,+}$ is the projection operator on the $G_i=1$ subspace. Note that, not only does $\rho_\ast$ not depend on the specific realization $\mathcal{Q}$, but it is also independent of $p_s$ and $p_u$, which means that as far as the expectation value of operators is concerned, the entire phase space looks the same.  Moreover, we show that the time it takes for the density matrix to reach the steady state is constant for $p_u=0$ while it is at most $\mathcal{O}(N)$ for $p_u\ne 0$.

For a general Pauli string operator $S$, the quantum channel corresponding to its measurement is given by
\begin{equation}\label{equ_channel}
  \mathcal{E}_S(\rho)=\Pi_{S,+}~\rho~\Pi_{S,+}+\Pi_{S,-}~\rho~\Pi_{S,-}
\end{equation}
where $\Pi_\pm$ denote the projectors onto $S=\pm 1$ subspaces, i.e.
\begin{equation}
  \Pi_{S,\pm}=\frac{1}{2}(\mathbbm{1}\pm S).
\end{equation}
By using the explicit form of the projectors $\Pi_{S,\pm}$, equation  \eqref{equ_channel} can be written as
\begin{equation}\label{eq_channel_stb}
  \mathcal{E}_S(\rho)=\frac{1}{2}\qty(\rho+S~\rho~S).
\end{equation}

Consider a mixed stabilizer state $\rho$
\begin{equation}
  \rho(\mathcal{G})=\frac{1}{2^N}\sum_{g\in\mathcal{G}}g,
\end{equation}
for a Pauli group $\mathcal{G}=\langle e_1,\cdot,e_n \rangle$ with $n\le N$ independent generators. According to equation  \eqref{eq_channel_stb}, under the measurement of a Pauli string $S$, we have
\begin{align}
  \mathcal{E}_{S}(\rho)&=\frac{1}{2^{N+1}}(\sum_{g\in \mathcal{G}} g + \sum_{g\in \mathcal{G}} S g S)\\
                       &=\frac{1}{2^{N}}\sum_{g\in C_\mathcal{G}(S)} g\\
                       &=\rho(C_\mathcal{G}(S)).\label{eq_rho1}
\end{align}
Here $C_\mathcal{G}(S)$ is the centralizer of $S$ in $\mathcal{G}$. If $S$ commutes with all elements in $\mathcal{G}$, clearly $C_\mathcal{G}(S)=\mathcal{G}$. Otherwise, without loss of generality, we can assume $S$ commutes with all generators of $\mathcal{G}$
except one of them, say $e_n$. Thus $C_\mathcal{G}(S)=\langle e_1,\cdots, e_{n-1}\rangle$.

The above analysis shows that for a mixed stabilizer state, whenever a Pauli string is measured, it either leaves the density matrix untouched or takes it to another mixed stabilizer state with one less generator, depending on whether the measured Pauli string commutes with the corresponding Pauli group or not.

In our case, the initial state of the system is given by
\begin{equation}\label{eq_init_dnst}
  \rho(\mathcal{G}_0)=\frac{1}{2^N}\sum_{g\in\mathcal{G}_0}g,\quad \mathcal{G}_0=\langle G_1,G_2,Z_2,\cdots,Z_{N-1}\rangle.
\end{equation}

Let us first consider the $p_u=0$ case. Based on the discussion above, it is clear that $Z_i$ measurements never change $\rho$ and thus can be ignored for our purpose. Each time a stabilizer $g_i$ is measured, $\rho(\mathcal{G})$ is transformed to $\rho(C_\mathcal{G}(g_i))$. Note that in general we have
\begin{equation}
  C_{C_\mathcal{G}(S_1)}(S_2)=C_\mathcal{G}(\{S_1,S_2\}) .
\end{equation}
Thus after all stabilizers $g_i$ have been measured at least once, the density matrix of the system would be given by:
\begin{align}
  \rho(C_{\mathcal{G}_0}(\{g_2,\cdots,g_{N-1}\}))&=\rho(\langle G_1,G_2 \rangle)\nonumber \\
                                                &=\frac{1}{2^{N-2}}\Pi_{G_1,+}\Pi_{G_2,+}=\rho_\ast.
\end{align}
Let $m_j$ denote the updating step at which $g_j$ is measured for the first time. It is easy to show that $\mathbb{E}[m_j]=(N-2)/p_t$, where $\mathbb{E}[X]$ denotes the expectation value of $X$. Therefore, the average time it takes for the system to reach the steady state $\rho_\ast$ would be
\begin{equation}
  \tau_\ast=\frac{1}{N}\mathbb{E}[\max_j(m_j)]=\frac{1}{N}\max_j(\mathbb{E}[m_j])=\mathcal{O}(1),
\end{equation}
where the pre-factor $1/N$ is there to convert updating steps to time steps.

Now consider the $p_u\ne 0$ case. Again, we start by the same initial density matrix given by equation  \eqref{eq_init_dnst}. Each time a measurement is performed, either $Z_i$ or $g_i$, the Pauli group associated with the density matrix of the system either remains the same or shrinks to one of its subgroups with one less generator, as was explained above. On the other hand, whenever a Clifford unitary $U$ is applied, it just transform $\rho(\mathcal{G})$ to $\rho(U^\dagger  \mathcal{G} U)$ with the same number of generators. Now, note that any Pauli group that commutes with every element in the set $\mathcal{M}=\{Z_i\}_{i=1}^N \cup \{g_i\}_{i=2}^{N-1}$ should be a subgroup of
$\langle G_1, G_2 \rangle$ (or the ones which are obtained by substituting $G_i$ with $-G_i$). Therefore, for any Pauli group $\mathcal{G}$ with more than two generators, there is at least one element of $\mathcal{M}$ that does not commute with $\mathcal{G}$. Ergo, at each updating step with probability of at least $\text{min}(p_s/N, p_t/(N-2))$, the Pauli group associated with the density matrix would shrink to a subgroup with one less generator, until only two generators $G_1$ and $G_2$ remain. Thus, on average, at most it takes $\mathcal{O}(N)$ updating steps until a stabilizer is measured which decreases the number of generators by one. Since we start with $N$ generators, the average time it takes to reach the steady state with only $G_1$ and $G_2$ as generators, i.e. $\rho_\ast$, would be:
\begin{equation}
  \tau_\ast \le \frac{1}{N}(N-2)\mathcal{O}(N)=\mathcal{O}(N).
\end{equation}

\section{String order parameters}\label{apx_stringOP}
The analysis in Supplementary Section  \ref{apx_std_rho} shows that the steady state density matrix is equal throughout the phase diagram and as such one could not detect the phase transition by averaging expectation values of operators over different realizations of the circuit. However, quantities like the ensemble average of the expectation value squared, which are not expressible in terms of the density matrix, could still be used as the order parameter.
As an example, consider the string order parameter introduced in Ref.\cite{perez2008string} to detect the SPT phase of the cluster model,
\begin{equation}
  s_{i,j}=\expval{X_{i-1}~Y_{i}~\prod_{k=i+1}^{j-1}Z_k ~ Y_{j}X_{j+1}}{\psi}=\expval{\prod_{k=i}^j g_k}{\psi}.
\end{equation}
Clearly, $s_{i,j}=\pm 1$ for symmetry invariant eigenstates of the cluster Hamiltonian while it is $0$ for any product states in the computational basis. If one considers the ensemble average of $s_{i,j}$ which could be expressed as $\overline{s_{i,j}}=\tr(\prod_{k=i}^j g_k\rho)$, it will zero even in the topological phase due to the cancelation between the terms with plus and minus signs. However this could be circumvented by considering the average over $s^2_{i,j}$ instead. In particular, the following order parameter, which can be viewed as the non-local analogue of the Edwards–Anderson glass-order parameter for SPT phases\cite{bahri2015localization,chandran2014many,sang2020measurement}, can be used to detect the phase transition:
\begin{equation}\label{eq_sop}
s=\frac{1}{N}\sum_{i<j} s_{i,j}^2.
\end{equation}
In the trivial phase, $s=0$ while in the SPT phase $s$ scales linearly with $N$. Fig.\ref{fig_s} shows the ensemble average of $s$, as a function of $p_s$ for different system sizes (all in $p_u=0$).
\begin{figure}
  \centering \includegraphics[width=\columnwidth]{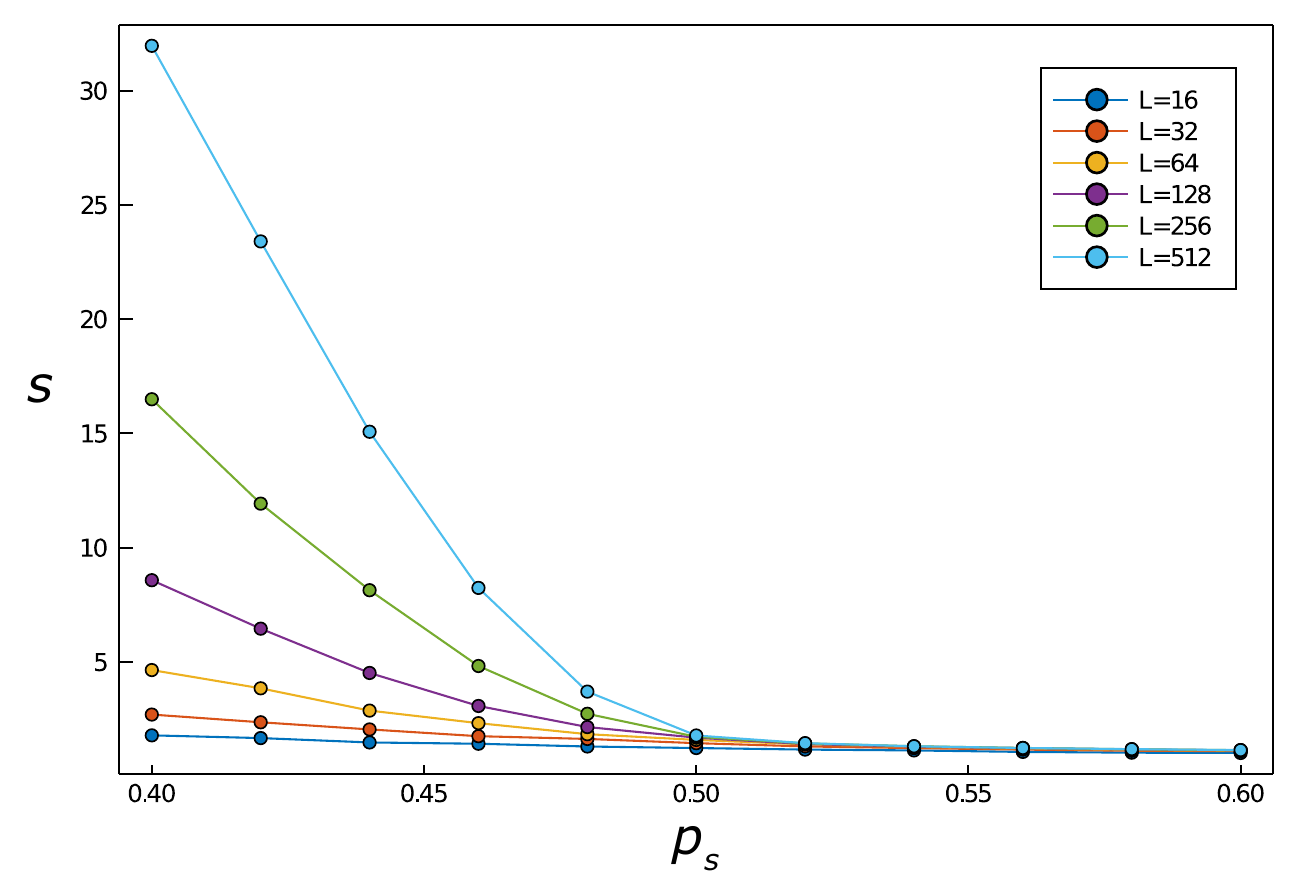}
  \caption{The string order parameter $s$ versus $p_s$ for different system sizes at $p_u=0$.}
  \label{fig_s}
\end{figure}

In contrast to $S_\text{topo}$, this order parameter can not distinguish between the volume-law phase and the trivial phase because it would be a constant in both phases\cite{sang2020measurement}.

 \section{Mapping the random circuit model to an unconventional percolation model}
 \label{apx_percolationMapMain}

Here we add some additional details regarding the the unconventional classical percolation model that corresponds to the random quantum circuit studied in the manuscript.

We focus on one of the two percolation lattices, say the odd sites' percolation lattice. Consider the $N/2 \times M$ square lattice of vertices. The connections on the $m$'th row and the connections between the rows $m$ and $m-1$, which are related to the measurement at step $m$, are determined as follows:
\begin{itemize}
  \item  With probability $\frac{1}{2}$ all vertical links between rows $m-1$ and $m$ are connected and all horizontal links on the $m$'th row are broken. This corresponds to the case where at step $m$, one of the even site operators is measured.
  \item With probability $\frac{1}{2}p_s$ a random vertical link between rows $m-1$ and $m$ is broken while all the others are connected. Also, all horizontal links on the $m$'th row are broken. This corresponds to the case where at step $m$, a single operator $Z_i$ is measured.
  \item  With probability $\frac{1}{2}(1-p_s)$ all vertical links between rows $m-1$ and $m$ as well as a random horizontal link on row $m$ are connected. All the other horizontal links on the $m$'th row are broken. This corresponds to the case where at step $m$, one of the stabilizer operators is measured.
\end{itemize}
As can be seen, the vertical links are mostly connected while the horizontal links are mostly broken. However, given that $M$ is equal to $2N^2$ ($2N$ time steps where each time step is consisted of $N$ updating steps), the lattice is elongated along the vertical direction.

As discussed in the main text, the entanglement entropy of the state at step $M$ can be read off entirely from the properties of the percolation lattice.

As is discussed in the main text the critical exponents $\nu$ and $z$ of this percolation model are the same as of the standard link percolation model on square lattice, thus suggesting it to be in the same universality class. Moreover, even the critical probability $p_c$ is the same. We also find that the coefficients of the entanglement entropy $a_x$ and $a_t$ computed in this model coincide with those of the circuit model defined in Supplementary Section  \ref{apx_layered_circuit_model}, which does map exactly to standard percolation on the square lattice.

\section{The duality map}\label{apx_U_proof}

 For simplicity, consider the system with periodic boundary conditions. Let us define the Clifford unitary $U_d$ such that for $i=1,\cdots,N$,
 \begin{align}
   {U_d}~X_i~U_d^\dagger &= X_i\label{eq_duality_X}\\
   {U_d}~Z_i~U_d^\dagger &= X_{i-1}~Z_i~X_{i+1}\label{eq_duality_g}
 \end{align}
 Note that under $U_d$, the stabilizer $g_i$ transforms as
 \begin{equation}\label{eq_duality_Z}
   {U_d}~g_i~ U_d^{\dagger}=Z_i.
 \end{equation}
 equation  \eqref{eq_duality_Z} and equation  \eqref{eq_duality_g} together show that the ensemble of random quantum circuits at $p_s$ (and $p_u= 0$) is mapped to the ensemble of random quantum circuit at $1-p_s$ (and $p_u = 0$) under $U_d$. However, the unitary $U_d$ is not local, i.e. it cannot be written as the tensor product of on-site unitaries and therefore does not keep the entanglement structure invariant. Nonetheless, it is clear from equation  \eqref{eq_duality_X} and equation  \eqref{eq_duality_g} that $U_d$ maps local stabilizers to local stabilizers. Since the entanglement in stabilizer states is related to the number of independent stabilizers that traverse the boundary of a region\cite{fattal2004entanglement}, one can still say that $U_d$ maps a state with the area-law entanglement to an area-law entangled state. Hence, if there exists a continuous phase transition which has logarithmic entanglement scaling, it has to occur at $p_s = p_c=1/2$. In the following this argument is made rigorous.

 Consider the stabilizer state $\ket{\psi}$ specified by the stabilizer set $\mathcal{S}=\{s_1,s_2,\cdots,s_N\}$ (see Supplementary Section  \ref{apx_stabilizer_formalism} for the notation). Let $A$ be the subset of qubits in the segment starting from the qubit $q_l$ and ending with the qubit $q_r$. Without loss of generality, we assume $s_1,\cdots,s_m$ generate the subgroup $\mathcal{G}_{A,\psi}$. Thus,
 \begin{equation}
   \mathcal{G}_{A,\psi}=\langle s_1,s_2,\cdots,s_m \rangle,\qquad \log_2 |\mathcal{G}_{A,\psi}|=m.
 \end{equation}
 We can also assume that at most $2$ stabilizers from the set $\{s_1,s_2,\cdots,s_m\}$ have non-trivial support on $q_l$ because if it is not the case, we can always use two of them to cancel the support of the others on $q_l$ by considering their multiplication with the rest. The same statement holds for $q_r$ as well.

 Now consider the state $\ket{\psi'}=U_d\ket{\psi}$ which corresponds to the stabilizer set $\mathcal{S'}=\{U_d s_1 U_d^\dagger, U_d s_2 U_d^\dagger,\cdots, U_d s_N U_d^\dagger\}$. Since $U_d$ moves the endpoint of a stabilizer by at most one site, at least $m-4$ out of $m$ stabilizers in $\{ U_d s_1 U_d^\dagger, U_d s_2 U_d^\dagger,\cdots, U_d s_m U_d^\dagger\}$ are still contained in $A$,
 which means $\log_2|\mathcal{G}_{A,\psi'}| \ge m-4$. Therefore,
 \begin{equation}\label{eq_ubound}
   S_{A}(U_d\ket{\psi})\le S_A(\ket{\psi})+4,
 \end{equation}
  which shows that if $\ket{\psi}$ has area-law entanglement, $U_d\ket{\psi}$ should have area-law entanglement as well.

  \section{Graphical representation of the state}
  \label{apx_graphical_rep}

 In this section we develop a graphical description to follow the system's state as it evolves under the random quantum circuit described in the main text for $p_u=0$. Moreover, this graphical representation provides the basic intuition behind the percolation mapping.

 The initial state of the system is given by the stabilizer set $\mathcal{S}_0=\{Z_1,Z_2,\cdots,Z_N\}$. At each step of the circuit, we measure either a stabilizer $g_i$ or a single qubit operator $Z_i$ on a qubit and update the stabilizer set accordingly. Let $\mathcal{S}_m$ denote the stabilizer set that corresponds to the system's state after $m$ updating steps.

 First, we prove the following Lemma:

 \begin{lemma}\label{lm_odd_even}
 $\mathcal{S}_m$ can be chosen such that each of its elements, up to a minus sign, is in one of the following forms:
 \begin{align}\label{eq_stab_strings}
  Z_{2j+1,2k+1}&\equiv\prod_{i=j}^k Z_{2i+1},\nonumber\\
  Z_{2j,2k}&\equiv \prod_{i=j}^k Z_{2i},\nonumber\\
  g_{2j+1,2k+1}&\equiv \prod_{i=j+1}^k g_{2i}=X_{2j+1} Z_{2j+2,2k} X_{2k+1},\nonumber\\
  g_{2j,2k}&\equiv\prod_{i=j}^{k-1} g_{2i+1}=X_{2j} Z_{2j+1,2k-1} X_{2k},
 \end{align}
 for some integers $j$ and $k$.
 \end{lemma}
 \begin{proof}
 We prove the lemma by induction. The claim is clearly true for $\mathcal{S}_0$. Assume it is true for $\mathcal{S}_m$.  First consider the case where we measure $Z_{2j+1}$ in the next step. We follow the procedure proscribed by the Gottesman-Knill theorem \cite{gottesman1998heisenberg,nielsen2002quantum} to obtain the stabilizer set $\mathcal{S}_{m+1}$. If $Z_{2j+1}$ commutes with every stabilizer in $\mathcal{S}_m$, then nothing happens by measuring $Z_{2j+1}$, hence $\mathcal{S}_{m+1}=\mathcal{S}_m$. So consider the case where some elements of $\mathcal{S}_m$ anti-commute with $Z_{2j+1}$.

 Any element of $\mathcal{S}_m$ that does not commute with $Z_{2j+1}$ has either the form $g_{2j+1,2k+1}$ or $g_{2k+1,2j+1}$ for some $k$.  If there is only one of them, then one only needs to replace it with $\pm Z_{2j+1}$ (with the sign chosen arbitrarily) to obtain $\mathcal{S}_{m+1}$. If there are more than one, we replace the first one with $ \pm Z_{2j+1}$, again with the sign chosen arbitrarily, and multiply the others with the stabilizer that was replaced, to get $\mathcal{S}_{m+1}$. In either cases, $\mathcal{S}_{m+1}$ will have the stated form.

 The other possibilities, i.e. measuring other operators at step $m+1$, can be treated similarly.
 \end{proof}

 Based on Lemma \ref{lm_odd_even}, we can use a diagrammatic notation to specify $\mathcal{S}_m$; we put $N$ dots along a line representing the qubits, as is shown in Fig.~\ref{fig_diagramatics}. Then, for every $g_{a,b}$ element in $\mathcal{S}_m$, we draw a line between sites $a,b$ from below and for every $Z_{a,b}$ element in $\mathcal{S}_m$ draw a line from above. Fig. \ref{fig_diagramatics}a and Fig. \ref{fig_diagramatics}b show the diagrams corresponding to $\mathcal{S}=\{Z_i\}_{i=1}^N$ and $\mathcal{S}=\{g_i\}_{i=1}^N$ respectively, with $g_1$ and $g_N$ defined as $g_1 \equiv G_1$ and $g_N \equiv G_2$.

 \begin{figure}
   \includegraphics[width=0.7\columnwidth]{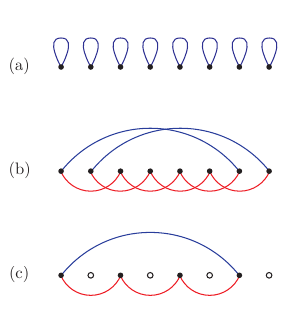}
   \caption{ The diagrammatic representation of the quantum state. \textbf{a}, The diagrammatic representation of the product state $\ket{0}^N$. \textbf{b}, The diagrammatic representation of the stabilizer state specified by $\mathcal{S}=\{g_i\}_{i=1}^N$. \textbf{c} the same as b, but including just the odd sites. }
   \label{fig_diagramatics}
 \end{figure}

 The form of the stabilizers listed in Lemma \ref{lm_odd_even} suggests a decomposition of the system into odd and even sites. Note that if we measure, for example, $Z_{2i+1}$, the only stabilizers that could be replaced are in the form $g_{2j+1,2k+1}$ while the $g_{2j,2k}$ stabilizers whose endpoints reside on even sites remain unchanged. Also, if one measures $g_{2i-1,2i+1}=g_{2i}$ whose ends points are on odd sites, the only stabilizers that could change have the form $Z_{2j+1,2k+1}$. So, if the stabilizer we are measuring has endpoints on odd sites, we only need to know about the stabilizers in $\mathcal{S}_m$ which also end on odd sites to find $\mathcal{S}_{m+1}$.  In other words, we can keep track of the set of stabilizers that start and end on odd sites, without knowing anything about the other stabilizers which start and end on the even sites and vice versa. This allows us to consider odd sites and even sites separately. Fig. \ref{fig_diagramatics}c shows the same state as in Fig. \ref{fig_diagramatics}b but restricted to odd sites only. For simplicity, we will only consider odd sites in what follows, while similar statements hold for even sites as well.

 \begin{figure}
   \includegraphics[width=\columnwidth]{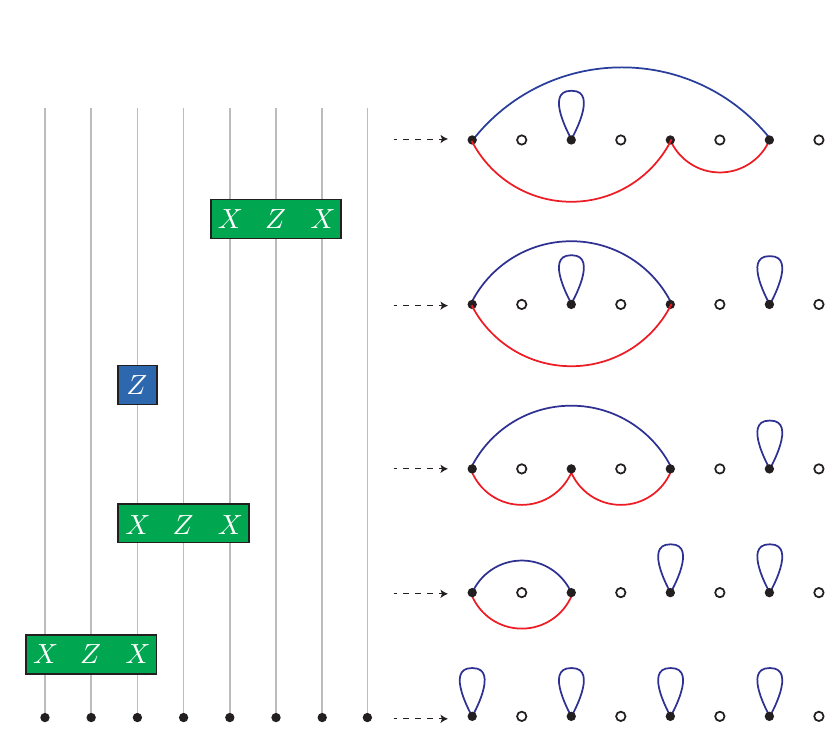}
   \caption{Step by step evolution of the system under the quantum circuit shown on the left. The diagrammatic representation of the system's state is shown on the right after each measurement.}
   \label{fig_diagramatics_clustering}
 \end{figure}
Using this diagrammatic formalism, it is easy to track $\mathcal{S}_m$. Fig.~\ref{fig_diagramatics_clustering} shows a typical quantum circuit and the step by step evolution of the system's stabilizer set using the diagrammatic notation developed above.

\section{Proof of Proposition 1}\label{apx_percolation_map}

 \begin{figure*}
   \includegraphics[width=\textwidth]{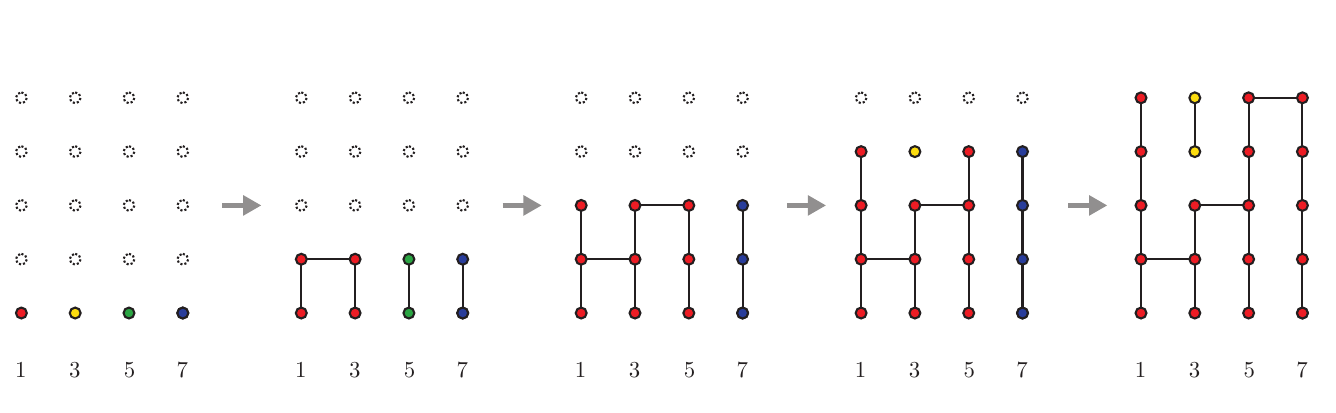}
   \caption{Step by step evolution of the system under the quantum circuit shown in Fig. \ref{fig_diagramatics_clustering} in the percolation picture for the odd sites. At each step, the sites with the same color on the last row represent an isolated SPT phase. There is an analogous diagram for the even sites. The two diagrams together fully specify the entanglement structure of the system.}
   \label{fig_percole_step_by_step}
 \end{figure*}
  We start by noting that for the circuit shown in Fig. \ref{fig_diagramatics_clustering}, the state of the system can always be described as a collection of isolated SPT states and decoupled qubits, as can be seen from the accompanying diagrammatic representation. For example, in the final state, qubits $1$, $5$ and $7$ form an isolated SPT state, while qubit $3$ is decoupled. This observation is indeed true in general. We start by putting forward a precise definition of an isolated SPT state and then show that there is an efficient description of $\mathcal{S}_m$ as a partition of $\{ 1,2, \cdots,N\}$.
 \begin{defini}
   Consider a set of numbers $A=\{q_i\}_{i=1}^n$, such that,
   \begin{equation}
     1\le q_1 < q_2 <\cdots<q_n \le N.
   \end{equation}
   Assume that all numbers are either odd or even. Define its associated stabilizer set, denoted by $\EuScript{S}(A)$, as
   \begin{equation}\label{eq_isolated_SPT_stabs}
     \EuScript{S}(A)=\{ g_{q_i,q_{i+1}} \}_{i=1}^{n-1} \cup \{ \prod_{i=1}^n Z_{q_i}\}.
   \end{equation}
   We call such a stabilizer set, an isolated SPT state.
 \end{defini}
 Note that the $g$ stabilizers in $\EuScript{S}(A)$ generate the set of all $g$ strings between any two points in $A$. Also, note that the stabilizers in equation  \eqref{eq_isolated_SPT_stabs} are the same as the ones appearing in Proposition 1.

 \begin{lemma}\label{lm_isolated_SPT}
   Let $\mathcal{S}_m$ denote the stabilizer set that corresponds to the system's state after $m$ updating steps. $\mathcal{S}_m$ can always be chosen such that, up to minus signs,
   \begin{equation}\label{eq_isolated_SPTs}
     \mathcal{S}_m=\cup_i \EuScript{S}(A_i),
   \end{equation}
   where $A_i$s correspond to a partition of the qubits into disjoint sets,
   \begin{equation}
     \cup_i A_i=\{1,2,\cdots,N\}.
   \end{equation}
   and $\EuScript{S}(A_i)$ denotes the isolated SPT state corresponding to subset $A_i$.
 \end{lemma}
 \begin{proof}
   We prove it by induction. It is obviously true for $\mathcal{S}_0$ with $A_i=\{i\}$ for $i=1,\cdots,N$.

   Now assume it is true for $\mathcal{S}_m$, so there exists a partition of qubits given by $\{1,2,\cdots,N\}=\cup_i A_i$ such that $\mathcal{S}_m=\cup_i \EuScript{S}(A_i)$. First consider the case in which a single qubit operator $Z_{2j+1}$ is measured in the next step. Suppose $2j+1$ is in subset $A_k$ for some $k$. Note that $Z_{2j+1}$ commutes with any element in $\EuScript{S}(A_{k'})$ with $k'\ne k$. If $A_k$ is the single element set $\{ 2j+1\}$
   (which means  $\EuScript{S}(A_{k})=\{Z_{2j+1}\}$)
   then $Z_{2j+1}$ is already in $S_m$ and thus $S_{m+1}=S_m$. If $A_k$ has more than one element, we will show that measuring $Z_{2j+1}$ corresponds to breaking $A_k$ to two subsets of $A_{k}\setminus \{ 2j+1 \}$ and $\{ 2j+1\}$.

  Note that $Z_{2j+1}$ anti-commutes only with the $g_{a,b}$ elements in $\EuScript{S}(A_k)$ where either $a$ or $b$ equals $2j+1$. If $2j+1$ is the smallest or largest number in $A_k$, there is only one such element and by measuring $Z_{2j+1}$, we just need to replace that element by $Z_{2j+1}$ (with an arbitrary sign) to get the updated stabilizer set $S_{m+1}$. If $2j+1$ is neither the smallest nor the largest number in $A_k$, there are two such elements, $g_{a,2j+1}$ and $g_{2j+1,b}$ for some odd numbers $a$ and $b$. Thus by measuring $Z_{2j+1}$, one is replaced by $Z_{2j+1}$ (with an arbitrary sign) and the other by $g_{a,2j+1}g_{2j+1,b}=g_{a,b}$ to get the updated stabilizer set. It is easy to verify that in both cases, $S_{m+1}$ is equivalent to the stabilizer set obtained by the union of isolated SPT states corresponding to the same partitioning as for $S_m$, but with $A_k$ broken to two sets of $A_{k}\setminus\{ 2j+1 \}$ and $\{ 2j+1\}$.

  Next consider the case where an stabilizer $g_{2j-1,2j+1}=g_{2j}$ is measured in the next step. If $2j-1$ and $2j+1$ belong to the same subset in the partition, nothing happens. If not, let say one belongs to $A_k$ and the other to $A_{k'}$, then $g_{2j-1,2j+1}$ anti-commute with the two $Z$ chains in $\EuScript{S}(A_k)$ and $\EuScript{S}(A_{k'})$ and commutes with everything else in $\mathcal{S}_m$. Therefore, by measuring $g_{2j-1,2j+1}$, we replace one of the $Z$ chains with $\pm g_{2j-1,2j+1}$ (with an arbitrary sign) and the other with the product of the two $Z$ chains, which is just the $Z$ chain over $A_k \cup A_{k'}$, to get $S_{m+1}$.
  It is straightforward to verify that $S_{m+1}$ is equivalent to the stabilizer set obtained by union of isolated SPT states corresponding to the same partitioning as for $S_m$, but by merging the two subsets $A_k$ and $A_{k'}$ into a single subset $A_k \cup A_{k'}$.
 \end{proof}

Based on Lemma \ref{lm_isolated_SPT}, there is a one-to-one mapping between partitions of $\{1,2,\cdots,N\}$ and the state of the system. Moreover, as can be seen from the Lemma's proof, the dynamics of the system can be translated into merging and splitting of the subsets.

Let us specify a partition by assigning unique colors to the qubits in the same subset. Then, whenever a $Z$ operator is measured, a new unique color should be assigned to the corresponding qubit to account for the new single element subset that is created in the new partitioning. On the other hand, when a $g$ operator with end points in different subsets is measured, the two subsets merge together which translates into assigning the same color to qubits in either one. The dynamics we have just described emerges naturally in the percolation model and thus can be used to map the quantum circuit to an instance of percolation on the square lattice.  We use a $N/2 \times M$ square lattice, where $M$ is the total number of updating steps. The $m$'th row of the lattice corresponds to the state of the system after the updating step $m$. We start by $N/2$ dots with distinct colors at the lowest row which corresponds to the initial product state. If $Z_{2i+1}$ is measured at step $m$, we leave the vertical link between $(2i+1,m)$ vertex and its history at $(2i+1,m-1)$ broken, while connecting all the other vertical links between the rows $m$ and $m-1$. By doing so, the $(2i+1,m)$ vertex gets a new color, while all the other vertices retain their color form the previous row, which agrees with the aforementioned splitting. On the  other hand, if a stabilizer $g_{2i-1,2i+1}$
 is measured at step $m$, we first connect all the vertical links between the rows $m$ and $m+1$, and then connect the vertices at $(2i-1,m)$ and $(2i+1,m)$ to enforce their colors to be the same, thus accounting for the aforementioned merging. Therefore, in each step, the colors of the last updated row can be used to find the partitioning of qubits mentioned in Lemma \ref{lm_isolated_SPT}, which completes the proof of Proposition 1.

As an example, Fig.~\ref{fig_percole_step_by_step} shows the step by step development of the circuit described in Fig.~\ref{fig_diagramatics_clustering}, in the percolation picture.

It is worth noting that the stabilizer set given in Lemma \ref{lm_isolated_SPT} is already in the clipped gauge (see Supplementary Section  \ref{apx_stabilizer_formalism} for the definition of clipped gauge). Therefore, the entanglement structure can be inferred readily from the percolation picture. In particular, the entanglement entropy $S(x)$ of the region $[1,x]$, is equal to the number of clusters with support on both inside and outside of the region $[1,x]$ on the top row of the percolation lattices. Such a quantity can be computed using the percolation CFT and the coefficient of the logarithm turns out to be:
\begin{equation}\label{eq_ax}
  a_x=\frac{\ln(2)\sqrt{3}}{2\pi}\simeq 0.191
\end{equation}
See for example equation (3) in Ref.\cite{cardy2001conformal} (see also \cite{cardy2000linking,nahum2020entanglement}). The $2\ln(2)$ discrepancy between equation \eqref{eq_ax} and equation (3) of Ref.\cite{cardy2001conformal} is due to the fact that we have two copies of percolation and the logarithm in our definition of entanglement entropy is in base $2$.

\section{Random measurement circuit corresponding to standard percolation}
\label{apx_layered_circuit_model}

Here we present a slightly different random circuit model than the one presented in the main text, which maps to the standard bond percolation on a square lattice. We define this circuit model using only the competing single qubit and stabilizer measurements, without including any unitary dynamics.

The random circuit model has a layered configuration of measurements, as follows. At alternating layers, we measure either only stabilizer operators or only single qubit operators, as shown in Fig. \ref{fig_circuit_layered}. Two consecutive layers of measurements correspond to a time step (here, there is no distinction between updating steps and time steps). In the layer where stabilizers are measured, for each $i$ we measure $g_i$ with probability $p_t$. In the layer where the single qubit operators are measured, for each $i$ we measure $Z_i$ with probability $1-p_t$. Fig.\ref{fig_circuit_layered} shows a typical realization of the circuit. Note that it is possible for two stabilizer measurements in the circuit to overlap with each other, but since the corresponding operators commute with each other, they can still be measured simultaneously.
\begin{figure}
  \includegraphics[width=0.5\columnwidth]{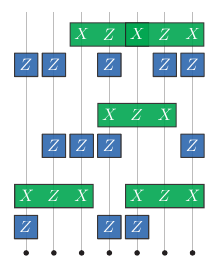}
  \caption{The random circuit which maps to the the standard bond percolation on a square lattice. On the last layer of the figure, two stabilizers which have overlap are measured. Since the stabilizers commute with each other, the order of the measurements does not matter. }
  \label{fig_circuit_layered}
\end{figure}

Note that the percolation mapping described in Supplementary Section  \ref{apx_percolation_map} is based on the specific form of the measurements involved rather than their layout on the circuit. Since the exact same set of measurements are involved in the quantum circuit described here, one can use the same rules, as is explained in the main text, to map this circuit model to two percolation models on square lattices. It is easy to verify each of the corresponding percolation models is indeed the standard bond percolation on an $N\times 2N$ square lattice where each bond is connected with probability $p_t$. We focus on the odd sites' lattice, for example, and include a horizontal bond with probability $p_t$, corresponding to a stabilizer measurement. A vertical bond is included with probability $p_t$, corresponding to the absence of a single qubit measurement.

Fig.~\ref{fig_layered_circuit_plots} is analogous to Fig. 3 in the main text, but for the random circuit model with the layered structure. The phase transition happens at $p_c=1/2$ and the critical exponent is found to be $\nu=4/3$, as is expected from the percolation mapping. We also find $a_x=a_t=0.20(1)$, the same values as the ones in the original circuit model in the main text.

\begin{figure}
  \includegraphics[width=\columnwidth]{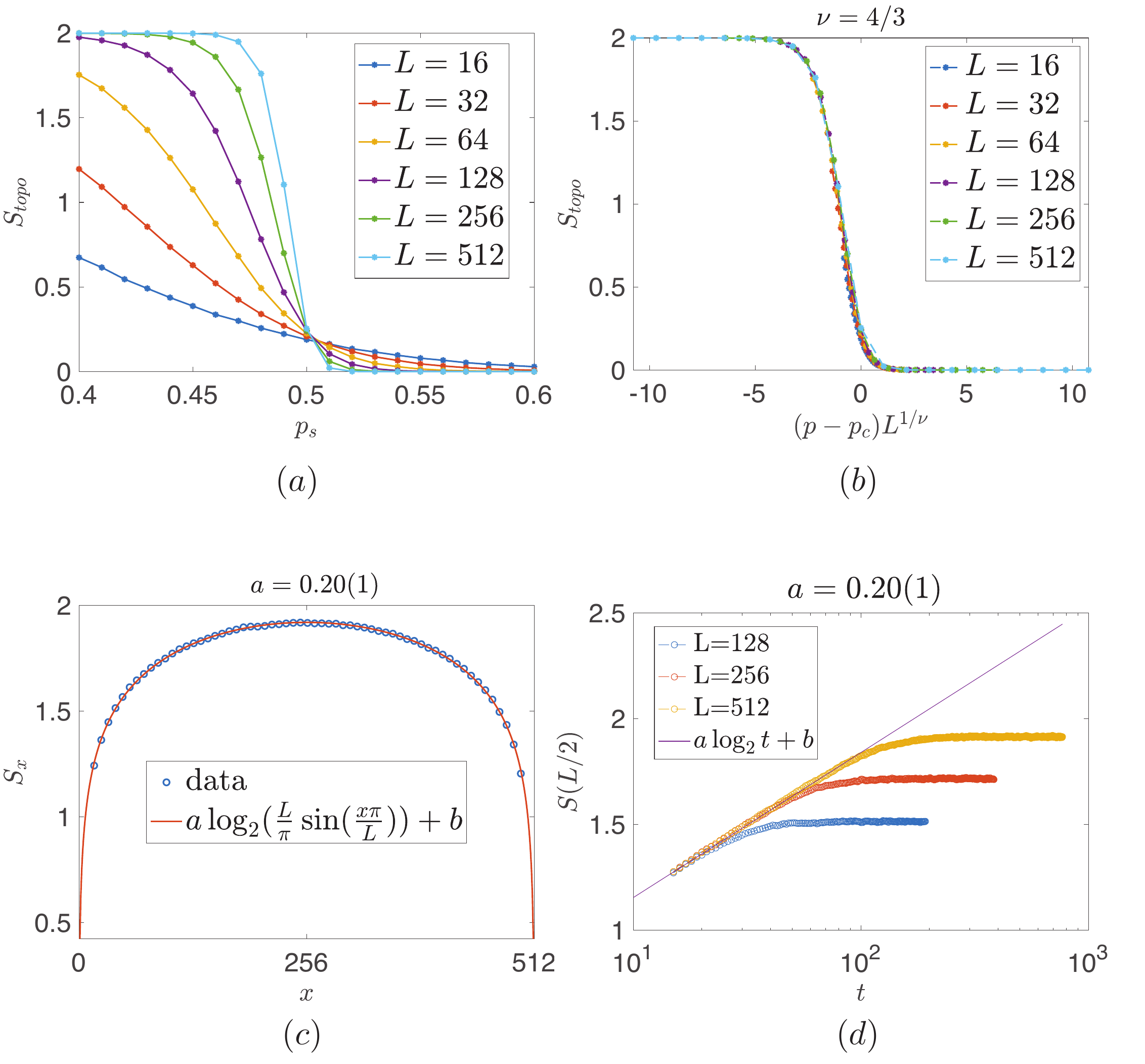}
  \caption{ Entanglement plots for the circuit with layered configuration. (a) Topological entanglement entropy $S_\text{topo}$ versus single qubit measurement probability which shows the same phase transition at $p_c=0.5$. (b) Scaling collapse of the data in (a) for $\nu=4/3$ using the same scaling as in equation (9) in the main text (c) The entanglement entropy of the $[0,x]$ segment of the chain in the steady state at $p=p_c$ and $L=512$. (d) The entanglement entropy of the half-chain versus time for $p_s=p_c$. }
  \label{fig_layered_circuit_plots}
\end{figure}

\section{Approaching the tricritical point}\label{apx_approaching_trc}

Figures \ref{fig_pu0102}a and b show the ancilla entropy $S_a$ measured $t=N$ time steps after it was entangled, versus $p_s$ for a fixed $p_u$ at $p_u=0.2$ and $p_u=0.1$ respectively (see Fig.4b in the main text for $p_u=0.3$ as well). The crossing points, mark the critical points which are shown as well in Fig.1b on the phase boundaries. As can be seen from the figures, there is clearly two separate area law phases with a volume law phase in between.

\begin{figure}
  \includegraphics[width=\columnwidth]{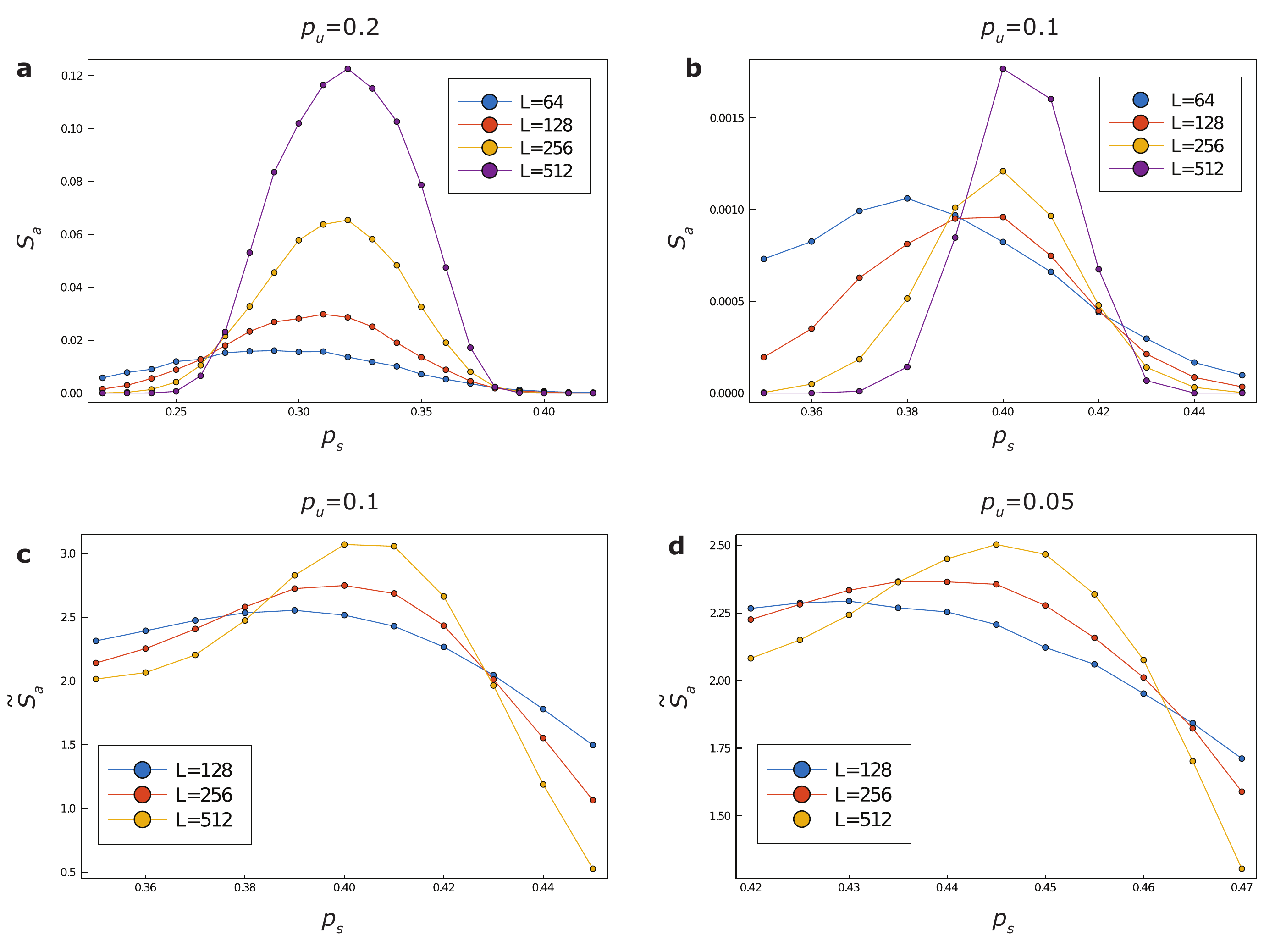}
  \caption{Ancilla order parameters in the vicinity of the tricritical point. \textbf{a} and \textbf{b}, The ancilla entropy $S_a$ versus $p_s$ measured $t=N$ time steps after being entangled, $p_u=0.2$ and $p_u=0.1$ respectively. \textbf{c} and \textbf{d}, The scrambled ancilla entropy $\tilde{S}_a$ versus $p_s$ measured $t=N/2$ time steps after being entangled, $p_u=0.1$ and $p_u=0.05$}
  \label{fig_pu0102}
\end{figure}
One thing to note is that as one probes lower values of $p_u$, the ancilla order parameter gets smaller. This in turns means one had to average over larger numbers of realizations to get get a result with decent noise. To see this, note that since the whole system is always in a stabilizer state, $S_\text{a}$ is either $0$ or $1$. and hence $\overline{S_a^2}=\overline S_a$, where $\overline{S}$ correspond to the ensemble average of $S$. Hence, the relative statistical error of $\overline{S_a}$ would be,
\begin{equation}
  \frac{\Delta \overline{S_a}}{\overline{S_a}}=\frac{\sqrt{(\overline{S_a^2}-\overline{S_a}^2)/N}}{\overline {S_a}}\simeq\frac{1}{\sqrt{N~ \overline{S_a}}}
\end{equation}
where we used $\overline{S_a}\ll 1$ in the last step. For example, the diagram in Fig.\ref{fig_pu0102}b is obtained by averaging over $10^6$ realizations (except $L=512$ size whihch is averaged over $4\times 10^5$ realization due to limited computational resources) while for Fig.\ref{fig_pu0102}a (and Fig.4 in the main text) averaging over $O(10^5)$ realizations results in a quite smooth curve.

On the other hand, evaluating the scrambled ancilla order parameter $\tilde{S}_a$ is more efficient and requires averaging over smaller number of realization to obtain  sufficiently smooth curves, and potentially could be used to prob the vicinity of the tricritical point. However,  it seems to be more sensitive to finite size effects (see also Fig.\ref{fig_pu03_zoomed}), and thus one needs to go to larger system sizes to be able to clearly distinguish different phases. Fig.\ref{fig_pu0102}c and d show the scrambled ancilla entropy measured $t=N/2$ time steps after scrambling at $p_u=0.1$ and $p_u=0.05$ respectively and are obtained via averaging over $10^4$ realizations. While Fig.\ref{fig_pu0102}c confirms the existence of the volume law phase at $p_u=0.1$ which is evident in Fig.\ref{fig_pu0102}b as well, Fig.\ref{fig_pu0102}d suggests that it survives down to $p_u=0.05$ as well but a more detailed study is needed to corroborate this claim.

  \section{Author Contribution}
  All authors contributed equally to this work.

  \section{competing interests}
  The authors declare no competing interests.
 \section{Data availability}
 The data plotted in the figures of this Article that support the findings of this study are available at \href{http://doi.org/10.5281/zenodo.4031884}{http://doi.org/10.5281/zenodo.4031884}
 \section{Code availability}
 The source codes used to run the simulations of the symmetric random quantum circuit studied in this Article are available at \href{http://doi.org/10.5281/zenodo.4031884}{http://doi.org/10.5281/zenodo.4031884}.

\clearpage
\widetext

\section{Supplemental Figures}\label{apx_additional_figs}

\begin{figure}[htbp]
  \includegraphics[width=\textwidth]{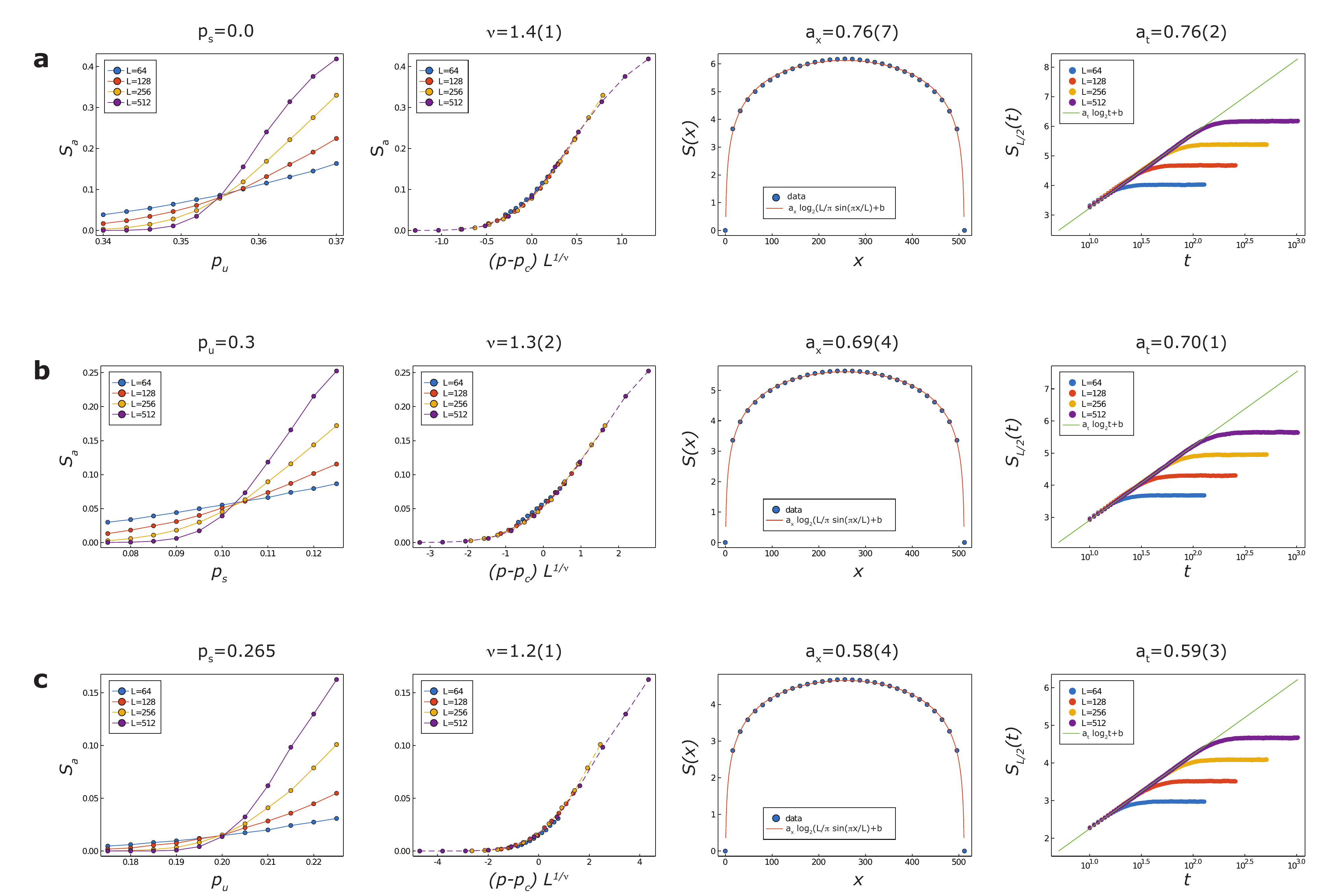}
  \caption{The transition from the SPT phase to the volume law entangled phase. The plots on each row corresponds to the same critical point. The ancilla order parameter $S_a$ is measured $t=N$ time steps after being entangled to the qubit chain. \textbf{a}, Plots corresponding to the critical point on the $p_u$ axis, with $p_s=0$ and $p_u=0.355(3)$.  \textbf{b}, Plots corresponding to the critical point at $p_s=0.102(3)$ and $p_u=0.3$. \textbf{c}, Plots corresponding to the critical point at $p_s=0.265$ and $p_u=0.201(4)$.   }
  \label{fig_SPTVL}
\end{figure}

\begin{figure}[htbp]b
  \includegraphics[width=\textwidth]{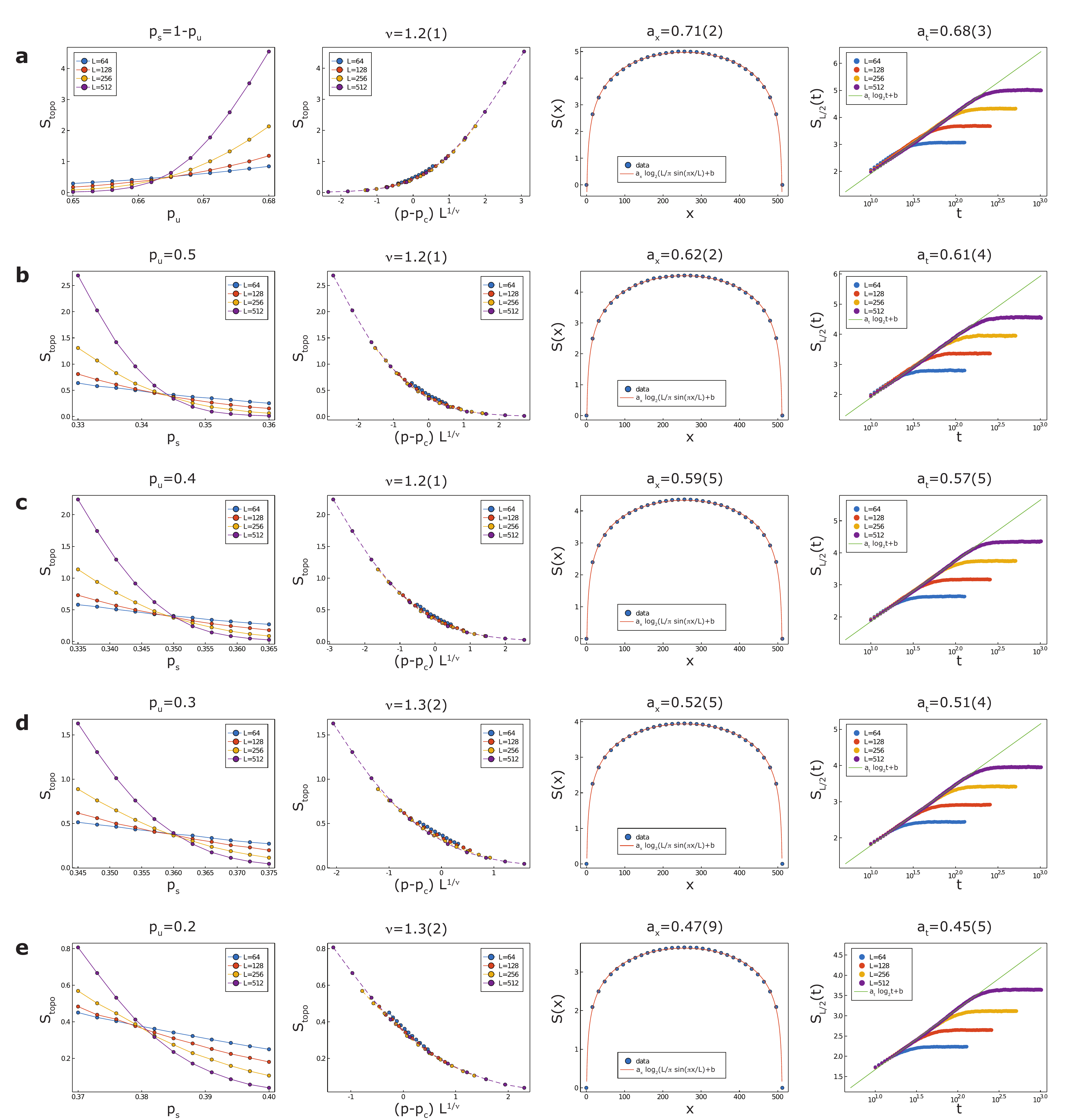}
  \caption{The transition from the volume law entangeld phase to the trivial phase. The plots on each row corresponds to the same critical point. \textbf{a}, Plots corresponding to the critical point on the $p_u+p_s=1$ line, with $p_s=0.337(4)$ and $p_u=0.663(4)$.  \textbf{b}, Plots corresponding to the critical point at $p_s=0.345(4)$ and $p_u=0.5$. \textbf{c}, Plots corresponding to the critical point at $p_s=0.351(3)$ and $p_u=0.4$. \textbf{d}, Plots corresponding to the critical point at $p_s=0.362(6)$ and $p_u=0.3$. \textbf{e}, Plots corresponding to the critical point at $p_s=0.381(4)$ and $p_u=0.2$.}
  \label{fig_VLTR}
\end{figure}

 \begin{figure}[htbp]
   \includegraphics[width=\textwidth]{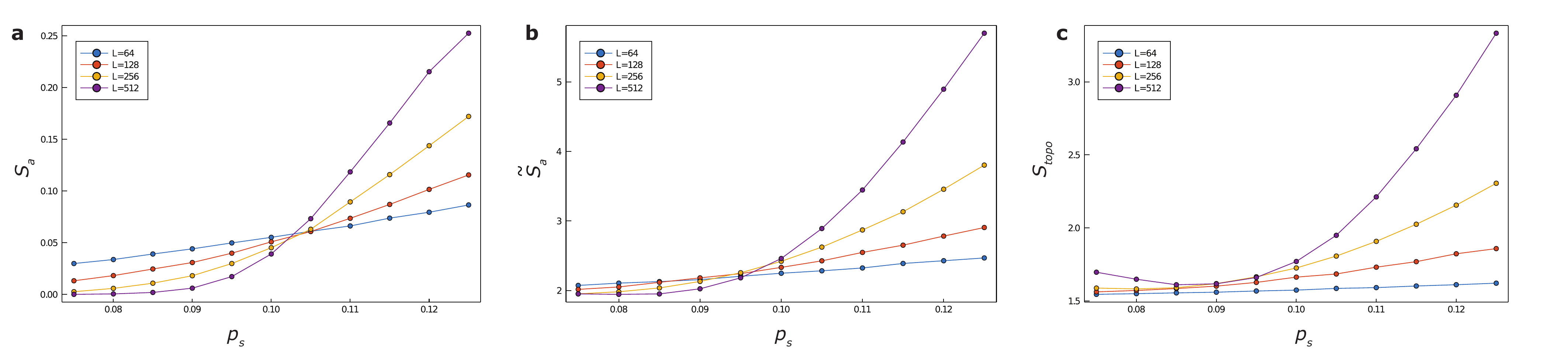}
   \caption{Three different order parameters at the SPT to volume law phase transition on $p_u=0.3$. \textbf{a}, the ancilla entropy $S_a$ measured at $t=N$. \textbf{b}, the scrambled ancilla entropy $\tilde{S}_a$ measured at $t=N/2$. \textbf{c} the topological entanglement entropy $S_\text{topo}$. Note that the ancilla order parameter $S_a$ provides the sharpest prob for the phase transition. }
   \label{fig_pu03_zoomed}
 \end{figure}

\end{document}